\documentclass[times, 10pt,twocolumn]{article} 
\usepackage{times}

\usepackage{amsmath}
\usepackage{amsthm} 
\usepackage{amssymb}
\usepackage{stmaryrd}
\usepackage{wrapfig}
\usepackage{xspace}
\usepackage{color}
\usepackage{gastex}
\usepackage{textcomp}

\usepackage[final]{listings}

\lstdefinelanguage{pseudo}{%
  morekeywords={Input, if, then, else, elseif, do, while, foreach, repeat, until, return, for,from, to, upto, case},
  morecomment=[l][keywordstyle]{//},
  escapeinside={(*@}{@*)},
  numberstyle=\tiny,
  stepnumber=1,
  numbersep=5pt,
  numberblanklines=false,
  numberfirstline=false,
  breaklines=true,
  numbers=left,
  xleftmargin=2.5ex,
  xrightmargin=2.4ex,
  breakindent=3ex,
  breakautoindent
}

\lstdefinestyle{pseudo}{%
  basicstyle=\rmfamily,
  emphstyle=\sffamily,
  mathescape=true,
  columns=flexible,
  language=pseudo
}

\newcommand{\m}{\text{\textminus}}

\newcommand{\last}{\textnormal{\text{last}}}

\newcommand{\outcome}{\textnormal{\textsf{Out}}}

\newcommand{\Outcome}{\textnormal{\textsf{Out}}}

\newcommand{\play}{\textnormal{\textsf{P}}}
\newcommand{\ut}{\textnormal{\textsf{u}}}
\newcommand{\rwc}{\textnormal{rwc}\xspace}
\newcommand{\twa}{\textnormal{TWA}\xspace}
\newcommand{\DAG}{\textnormal{DAG}\xspace}
\newcommand{\RMP}{\textnormal{RMP}\xspace}

\newcommand{\target}{\textnormal{\textsf{C}}}

\newcommand{\br}{\textnormal{\textsf{br}}}
\newcommand{\reg}{\textnormal{\textsf{reg}}}

\newcommand{\best}{\textnormal{\textsf{best}}}

\newcommand{\minmax}{\textnormal{\textsf{minmax}}\xspace}

\newcommand{\ba}{\textnormal{\textsf{ba}}}

\newcommand{\proj}{\textnormal{\textsf{proj}}}

\newtheorem{proposition}{Proposition}
\newtheorem{definition}{Definition}
\newtheorem{lemma}{Lemma}
\newtheorem{theorem}{Theorem}

\newtheorem{example}{Example}

\begin{document}

\title{Iterated Regret Minimization in Game Graphs}

\author{Emmanuel Filiot \qquad Tristan Le Gall \qquad Jean-Fran\c{c}ois Raskin\\
 efiliot@ulb.ac.be\qquad
tlegall@ulb.ac.be\qquad
jraskin@ulb.ac.be \\
Universit\'e Libre de Bruxelles\\
}

\maketitle
\thispagestyle{empty}

\begin{abstract}
Iterated regret minimization has been
introduced recently by J.Y. Halpern and R. Pass in classical
strategic games. For many games of interest, this new solution concept
provides solutions that are judged more reasonable than solutions offered by traditional game concepts -- such as \textit{Nash equilibrium}
--. Although computing iterated regret on explicit matrix game is conceptually and computationally easy, nothing is known about computing 
the iterated regret on games whose matrices are defined implicitly using game tree, game \DAG or, more generally game graphs.
In this paper, we investigate iterated regret minimization for infinite duration two-player
quantitative non-zero sum games played on graphs. 

We consider reachability objectives that are
not necessarily antagonist. Edges are weighted by integers -- one for
each player --, and the payoffs are defined by the sum of the weights
along the paths. Depending on the class of graphs, we give either
polynomial or pseudo-polynomial time algorithms to compute a strategy that
minimizes the regret for a fixed player. We finally give algorithms
to compute the strategies of the two players that minimize the iterated
regret for trees, and for graphs with strictly positive weights only.
\end{abstract}


\vspace{-7mm}
\section{Introduction}
\vspace{-3mm}

The analysis of complex interactive systems like embedded systems or distributed
systems is a major challenge of computer aided verification. Zero-sum
games on graphs provide a good framework to model interactions
between a component and an environment as they are strictly
competitive. However in the context of
modern interactive systems, several components may interact and be
controlled independently. Non-zero sum games on graphs
are more accurate to model such systems, as the objectives are not
necessarily antagonist. There are initial results in this area
but a large number of questions are open. In this paper, we
adapt to game graphs a new solution concept of non-zero sum games
initially defined for strategic games.

In~\cite{HalpernPass09}, J.Y. Halpern and R. Pass defined the notion of \emph{iterated regret minimization}. This solution concept assumes that instead of trying to minimize what she has to pay, each player tries to minimize her {\em regret}. The regret is informally defined as the difference between what a player actually pays and what she could have payed if she knew the strategy chosen by the other player.  More formally, if $\ut_1(\lambda_1,\lambda_2)$ represents what Player 1 pays when the pair of strategies $(\lambda_1,\lambda_2)$ is played, $\reg_1(\lambda_1,\lambda_2) = \ut_1(\lambda_1,\lambda_2) - \min_{\lambda'_1} \ut_1(\lambda'_1,\lambda_2)$. 

Let us illustrate this on an example. Consider the strategic game defined by the matrix of figure~\ref{matrix-ex}. In the game underlying this matrix, Player 1 has two strategies $A_1$ and $B_1$ and Player 2 has two strategies $A_2$ and $B_2$. The two players choose a strategy at the same time and the pairs of strategies define what the two players have to pay\footnote{We could have considered rewards instead of penalties, everything is symmetrical.}. The regret of playing $A_1$ for Player 1 if Player 2 plays $A_2$ is equal to $1$ because $\ut_1(A_1,A_2)$ is $2$ when $\ut_1(B_1,A_2)$  is $1$. Knowing that Player 2 plays $A_2$, Player 1 should have played $B_1$.  

\begin{figure}
\label{matrix-ex}
\[
 \begin{array}{|c|c|c|}
    \hline
    & A_2 & B_2 \\
    \hline
    A_1 & (2,1) & (3,4) \\
    \hline
    B_1 & (1,2) & (4,3) \\
    \hline
  \end{array}
\]
\caption{A strategic game with explicit penalty matrix.}
\vspace{-7mm}
\end{figure}

As Players have to choose strategies before knowing how the adversary will play, we associate a regret with each strategy as follows. The regret of a strategy $\lambda_1$ of Player 1 is~: $\reg_1(\lambda_1) \; = \; \max_{\lambda_2} \reg_1(\lambda_1,\lambda_2)$.  In the example, the regret attached to strategy $A_1$ is equal to $1$, because when Player 2 plays $A_2$, Player 1 regret is $1$, and when Player 2 plays $B_2$ her regret is $0$. A rational player should minimize her regret. The regret for Player 1 is thus defined as $\reg_1 \; = \; \min_{\lambda_1}  \reg_1(\lambda_1)$, summarizing, we get $\reg_1 \; = \; \min_{\lambda_1} \max_{\lambda_2}  ( \ut_1(\lambda_1,\lambda_2) - \min_{\lambda'_1} \ut_1(\lambda'_1,\lambda_2))$. A symmetrical definition can be given for Player $2$'s regret.

Let us come back to the example. The regret attached to strategy $B_1$ is equal to $1$. So the two strategies of Player 1 are equivalent w.r.t. regret minimization. On the other hand, for Player 2, the regret of $A_2$ equals $0$, and the regret of $B_2$ equals $3$.  So, if Player 1 makes the hypothesis that Player 2 is trying to minimize her regret, then she must conclude that Player 2 will play $A_2$. Knowing that, Player 1 recomputes her regret for each action, and in this case, the regret of action $A_1$ is $1$ while the regret of $B_1$ is $0$. So rational players minimizing their regret should end up playing the pairs $(B_1,A_2)$ in this game.

Reasoning on rationality is formalized by Halpern and Pass by introducing a {\em delete operator} that erases strictly dominated strategies. This operator takes sets of strategies $(\Lambda_1,\Lambda_2)$ for each player and returns  $D(\Lambda_1,\Lambda_2)=(\Lambda_1',\Lambda_2')$ the strategies that minimize regret. Then $D(\Lambda_1',\Lambda_2')$ returns the strategies that minimize regret under the hypothesis that adversaries minimize their regret i.e., choose their strategies in $\Lambda_1'$ and $\Lambda_2'$ respectively. In the case of finite matrix games, this operator is monotone and converges on the strategies that minimize regrets for the two players making the assumption of rationality of the other player.
 
In this paper, we consider games where the matrix is not given explicitly but defined implicitly by a game graph. More precisely, we consider graphs where vertices are partitioned into vertices that belong to Player 1 and vertices that belong to Player 2. Each edge is annotated by a penalty for Player 1 and one for Player 2. Additionally, there are two designated sets of vertices, one that Player 1 tries to reach and the other one that Player 2 tries to reach. The game starts in the initial vertex of the graph and is played for an infinite number of rounds as follows.  In each round, the Player who owns the vertex on which the pebble is placed moves the the pebble to an adjacent vertex using an edge of the graph, and a new round starts. The infinite plays generate an infinite sequence of vertices and the amount that the players have to pay are computed as follows. Player 1 pays $+\infty$ if the sequence does not reach the target set assigned to Player 1, otherwise she pays the sum of edge costs assigned to her on the prefix up to the first visit to her target set. The amount to pay for Player 2 is defined symmetrically. Strategies in such games are functions from the set of histories of plays (sequences of visited vertices) to edges (choice of moves for the pebble).

Let us consider the game graph of Fig.~\ref{centri-ex}. This is a
formalization of the so-called Centipede game~\cite{Rosenthal82} in our game
graphs. We have considered a 5-round variant here, this game can be
generalized to any number of rounds. Initially, the pebble is on
vertex $A$. Player 1 owns the circle vertices and Player 2 owns the
square vertices. The target objective for the two players is the
same: they both want to reach vertex $S$. At each round, one of the players
has to choose either to stop the game and reach the target, or to let
the game continue for at least an additional round. The penalties attached to
edges are given as pairs of integers (the first for Player 1 and the
second for Player 2). Strategies here are as follows. For each circle
vertex, Player 1 must decide either to continue or to go to the target
$S$, and symmetrically for Player 2. It can be shown (and computed by
our algorithms) that the strategy of Player 1 that survives iterated
regret minimization is the strategy that stops the game only in
position $E$ and the strategy for Player 2 is the strategy that
continue the game to vertex $D$. This pair of strategy has a penalty of
$(1,3)$. This is an interesting and rather nice joint
behavior of the two players in comparison of what Nash equilibrium is
predicting for this example. Indeed, the only Nash
equilibrium\footnote{A Nash equilibrium is a pair of strategies where no player has an incentive to change her strategy if the other
  player keeps playing her strategy} in that game is the pair of strategies where the two players decide to stop directly the game and so they have to pay $(5,7)$. This is a 5-round example but the difference between the penalty of the Nash equilibrium and the iterated regret grows as the number of rounds increases.

\vspace{-3mm}
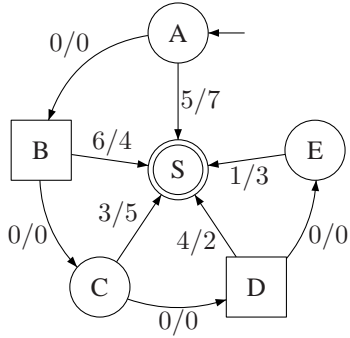
\begin{figure}[h]
\label{centri-ex}
\centering{
  \unitlength=0.65mm

\begin{picture}(84,71)(0,-71)
\node[NLangle=0.0,Nmr=0.0](n9)(20.0,-32.0){B}

\node[NLangle=0.0](n10)(32.0,-60.0){C}

\node[NLangle=0.0,Nmr=0.0](n11)(64.0,-60.0){D}

\node[NLangle=0.0](n12)(76.0,-32.0){E}

\node[NLangle=0.0,Nmarks=r](n13)(48.0,-36.0){S}

\drawedge[ELdist=0.3,ELside=r,curvedepth=-4.48](n9,n10){$0/0$}

\drawedge[ELdist=0.3,ELside=r,curvedepth=-4.23](n10,n11){$0/0$}

\drawedge[ELdist=0.3,ELside=r,curvedepth=-4.45](n11,n12){$0/0$}

\drawedge[ELdist=0.3](n9,n13){$6/4$}

\drawedge[ELdist=0.3](n10,n13){$3/5$}

\drawedge[ELdist=0.3](n11,n13){$4/2$}

\drawedge[ELdist=0.3](n12,n13){$1/3$}

\node[NLangle=0.0,ilength=7.41,iangle=0.0,Nmarks=i](n8)(48.0,-8.0){A}

\drawedge[ELdist=0.3,ELside=r,curvedepth=-7.7](n8,n9){$0/0$}

\drawedge[ELdist=0.5](n8,n13){$5/7$}

\end{picture}
\vspace{-4.5mm}
  \caption{\label{fig:centipede} Centipede Game}
}
\vspace{-6mm}
\end{figure}

\paragraph{Contributions}
We investigate iterated regret minimization for
infinite duration two-player quantitative non-zero sum games played on
graphs. We focus on reachability objectives that are not necessarily
antagonist.

We first consider target-weighted arenas, where the payoff function is
defined for each state of the objectives. We give a PTIME algorithm to
compute the regret by reduction to a min-max game.

We then consider edge-weighted arenas. Each edge is labeled by a pair of integers -- one for each player
--, and the payoffs are defined by the sum of the weights along the
path until the first visit to an objective. We give a
pseudo-PTIME algorithm to compute the regret in an edge-weighted
arena, by reduction to a target-weighted arena.

We also study the problem of iterated regret minimization. 
We provide a \emph{delete operator} that removes strictly dominated
strategies. We show how to compute the effect of iterating this
operator on tree arenas and strictly positive
edge-weighted arenas. In the first case, we provide a quadratic time algorithm and in the second case,
a pseudo-exponential time algorithm.

\vspace{-2mm}

\paragraph{Related works}
Several notions of equilibria have been proposed
in the literature for reasoning on 2-players non-zero-sum games,
for instance \emph{Nash equilibrium}, \emph{sequential equilibrium},
\emph{perfect equilibrium} - see~\cite{Osborne94} for an overview.
Those equilibria formalize notions of rational behavior by defining
optimality criteria for pairs of strategies. As we have seen in the
Centipede game example~\cite{Rosenthal82}, or as it can be shown
for other examples like the \emph{Traveller's dilemma}~\cite{Basu94},
Nash equilibria sometimes suggest pairs of strategies that are rejected by common sense. Regret
minimization is an alternative solution concept that sometimes
proposes more intuitive solutions and requires more cooperation between
players. Recently, non-zero sum games played on graphs have attracted a
lot of attention. There have been several papers that study Nash
equilibria or particular classes of Nash equilibria 
\cite{Ummels08,Chatterjee04gameswith,Brihaye10,Kupf10}.

\textit{Proofs that are sketched or omitted in the paper are given in
  Appendix}.








\section{Weighted Games and Regret}

Given a cartesian product $A\times B$ of two sets, we
denote by $\proj_i$ the $i$-th projection, $i=1,2$. It is naturally
extended to sequence of elements of $A\times B$ by 
$\proj_i(c_1\dots c_n) = \proj_i(c_1)\dots\proj_i(c_n)$. 
For all $k\in\mathbb{N}$, we let $[k] = \{0,\dots,k\}$.


\paragraph{Reachability Games} Turn-based two-player games are played on game arenas by
two players. A (finite) \textit{game arena} is a tuple $G =
(S = S_1\uplus S_2,s_0,T)$ where $S_1,S_2$
are finite disjoint sets of player positions ($S_1$ for Player 1
and $S_2$ for Player $2$), $s_0\in S_1$ is the initial
position, and $T\subseteq S\times S$ is the
transition relation. A \textit{finite play} on $G$ of length $n$ is
a finite word $\pi = \pi_0\pi_1\dots\pi_n\in S^*$
such that $\pi_0 = s_0$ and for all $i=0,\dots,n-1$,
$(\pi_i,\pi_{i+1})\in T$. Infinite plays are defined
similarly. We denote by $\play_f(G)$ (resp. $\play_{\infty}(G)$) the
set of finite (resp. infinite) plays on $G$, and we let $\play(G) =
\play_f(G)\cup \play_{\infty}(G)$. For any node $s\in
S$, we denote by $(G,s)$ the arena $G$ where the initial position is
$s$.

Let $i\in\{1,2\}$. We let $\m i = 1$ if $i=2$ and $\m i=2$ if $i=1$.
A \textit{strategy} $\lambda_i : \play_f(G)\rightarrow S\cup \{\perp\}$ for Player $i$
is a mapping that maps any finite play $\pi$ whose
last position -- denoted $\last(\pi)$ -- is in $S_i$ to $\perp$ if
there is no outgoing edge from $\last(\pi)$, and to 
a position $s$ such that $(\last(\pi),s)\in T$ otherwise. The set of strategies of Player $i$ in
$G$ is denoted by $\Lambda_i(G)$. Given a strategy $\lambda_{\m i}\in\Lambda_{-i}(G)$,
the \textit{outcome} $\outcome^G(\lambda_i,\lambda_{\m i})$ is
a play $\pi = \pi_0\ldots \pi_n\ldots$ such that $(i)$ $\pi_0 = s_0$, $(ii)$ if $\pi$ is finite, 
then there is not outgoing edge from $\last(\pi)$, and $(iii)$ 
for all $0\leq j \leq |\pi|$ and all $\kappa=1,2$, if $\pi_j\in S_\kappa$, then
$\pi_{j+1} = \lambda_\kappa(\pi_0\dots \pi_j)$. We also define $\outcome^G(\lambda_i) =
\{ \outcome^G(\lambda_i,\lambda_{\m i})\ |\ \lambda_{\m i}\in
\Lambda_{\m i}(G)\}$.

A strategy $\lambda_i$ is \textit{memoryless} if for all play $\pi\in\play_f(G)$, 
$\lambda_i(\pi)$ only depends on $\last(\pi)$. Thus $\lambda_i$ can be seen as a function $S_i\mapsto S\cup \{\perp\}$.
It is \textit{finite-memory} if 
$\lambda_i(\pi)$ only depends on $\last(\pi)$ and on some state of a finite state set.
We refer the reader to \cite{GraedelTW02} for formal definitions.


A \textit{reachability winning condition (\rwc for short)} for Player $i$ is given by
a subset of positions $\target_i\subseteq S$ -- called the \textit{target set} --. A play $\pi\in \play(G)$ is winning 
for Player $i$ if some position of $\pi$ is in $\target_i$.
A strategy  $\lambda_i$ for Player
$i$ is \textit{winning} if all the plays of $\Outcome^G(\lambda_i)$ are winning.
In this paper, we often consider two target sets $\target_1,\target_2$ for Player 1 and 2 respectively.
We write $(S_1,S_2,s_0,T,\target_1,\target_2)$ to denote
the game arena $G$ extended with those target sets. Finally, let
 $\lambda_i\in\Lambda_i(G)$ be a winning strategy for Player $i$ and
$\lambda_{\m i}\in\Lambda_{\m i}(G)$. Let $\pi_0\pi_1\dots\in\play(G)$
be the outcome of $(\lambda_i,\lambda_{\m i})$. The outcome of
$(\lambda_i,\lambda_{\m i})$ up to $\target_i$ is defined by
$\outcome^{G,\target_i}(\lambda_i,\lambda_{\m i}) = \pi_0\dots \pi_n$
such that $n = \min \{ j\ |\ \pi_j\in \target_i\}$. We also extend
this notation to sets of plays $\outcome^{G,\target_i}(\lambda_i)$ naturally.


\paragraph{Weighted Games} We add weights on edges of arenas 
and include the target sets. A (finite) \textit{weighted game arena}
is a tuple $G = (S = S_1\uplus S_2,s_0,T,\mu_1,\mu_2,\target_1,\target_2)$ where
$(S,s_0,T)$ is a game arena, for all $i=1,2$, $\mu_i : T\rightarrow
\mathbb{N}$ is a weigth function for Player $i$ and $\target_i$ its
target set. We let $M_i^G$ be the maximal weight of Player $i$, i.e. 
$M_i^G = \max_{e\in T} \mu_i(e)$ and $M^G = \max (M_1^G,M_2^G)$.

$G$ is a \textit{target-weighted arena} (\twa for short) if only 
the edges leading to a target node are weighted by strictly positive
integers, and any two edges leading to the same node carry the same
weight. Formally, for all $(s,s')\in T$, if
$s'\not\in\target_i$, then $\mu_i(s,s') = 0$, otherwise for all $(s'',s')\in T$, $\mu_i(s,s') =
\mu_i(s'',s')$. Thus for target-weighted arenas, we assume
in the sequel that the weight functions map $\target_i$ to
$\mathbb{N}$.

Let $\pi = \pi_0\pi_1\dots \pi_n$ be a finite play in $G$. We extend the weight
functions to finite plays, so that for all $i=1,2$, $\mu_i(\pi) =
\sum_{j=0}^{n-1} \mu_i(\pi_j,\pi_{j+1})$. The \textit{utility} $\ut^G_i(\pi)$ of $\pi$ (for Player
$i$) is $+\infty$ if $\pi$ is not winning for Player $i$, and the sum of
the weights occuring along the edges defined by $\pi$ until the first
visit to a target position otherwise. Formally:


$$
\ut^G_i(\pi)\ =\ 
\left\{\begin{array}{llllllll}
+\infty \qquad\quad\,\,\,\, \text{ if } \pi \text{ is not winning for Player $i$} \\
\sum_{j=0}^{\min \{ k\ |\ \pi_k\in \target_i\}-1} 
\mu_i(\pi_j,\pi_{j+1}) \quad \text{ otherwise}
\end{array}\right.
$$ 


We extend this notion to the utility of two strategies
$\lambda_1,\lambda_2$ of Player 1 and 2 respectively:


$$
\ut^G_i(\lambda_1,\lambda_2)\ =\ \ut^G_i(\outcome^G(\lambda_1,\lambda_2))
$$


Let $i\in\{1,2\}$. Given a strategy $\lambda_i\in\Lambda_i(G)$, the \textit{best response}
of Player $\m i$ to $\lambda_i$, denoted by $\br^G_{\m i}(\lambda_i)$, is
the least utility Player $\m i$ can achieve against 
$\lambda_i$. Formally:


$$
\br^G_{\m i}(\lambda_i)\ =\ \min_{\lambda_{\m i}\in\Lambda_{\m i}(G)} \ut_{\m i}^G(\lambda_i,\lambda_{\m i})
$$

\paragraph{Regret} Let $i\in\{1,2\}$ and let 
$\lambda_1,\lambda_2$ be two strategies of Player 1 and 2 respectively. 
The \textit{regret} of Player $i$ is the difference 
between the utility Player $i$ achieves and the best response
to the strategy of Player $\m i$. Formally:


$$
\reg_i^G(\lambda_i,\lambda_{\m i})\ =\ \ut^G_i(\lambda_i,\lambda_{\m i}) - \br_i^G(\lambda_{\m i})
$$


Note that $\reg_i^G(\lambda_i,\lambda_{\m i})\geq 0$, since
$\br_i^G(\lambda_{\m i})\leq \ut^G_i(\lambda_i,\lambda_{\m i})$.
The regret of a strategy $\lambda_i$ for Player $i$ is the maximal regret she gets
for all strategies of Player $\m i$:


$$
\reg_i^G(\lambda_i)\ =\ \max_{\lambda_{\m i}\in \Lambda_{\m i}(G)} \reg_i^G(\lambda_i,\lambda_{\m i})
$$


Finally, the regret of Player $i$ in $G$ is the minimal regret she can
achieve:


$$
\reg^G_i = \min_{\lambda_i\in\Lambda_i(G)} \reg^G_i(\lambda_i)
$$


We let $+\infty - (+\infty) = +\infty$.

\begin{proposition}
For all $i=1,2$, $\reg_i^G<+\infty$ iff Player $i$ has a winning
strategy.
\end{proposition}

\begin{proof}
  If Player $i$ has no winning strategy, then for all
  $\lambda_i\in\Lambda_i(G)$, there is $\lambda_{\m i}\in\Lambda_{\m i}(G)$ s.t. 
  $\ut_i^G(\lambda_i,\lambda_{\m i}) = +\infty$. Thus $\reg_i^G(\lambda_i,\lambda_{\m i}) = +\infty$. Therefore $\reg_i^G = +\infty$.

  \noindent If Player $i$ has a winning strategy $\lambda_i$, then for all
  $\lambda_{\m i}\in\Lambda_{\m i}(G)$, $\ut_i^G(\lambda_i,\lambda_{\m
    i}) < +\infty$ and $\br_i^G(\lambda_{\m i}) \leq
  \ut_i^G(\lambda_i,\lambda_{\m i}) < +\infty$. Thus
  $\reg_i^G \leq \reg_i^G(\lambda_i) < +\infty$. 
\end{proof}

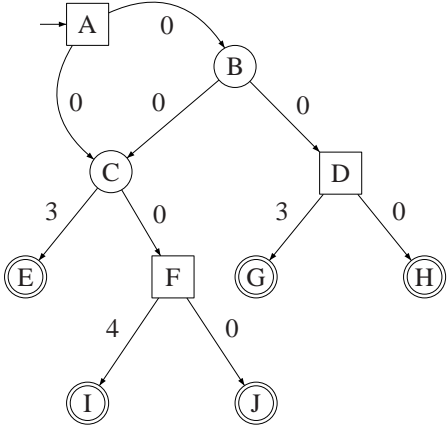
\begin{figure}[h]
  \vspace{-25mm}
  \unitlength=0.7mm



\gasset{linewidth=0.14,Nw=8.0,Nh=8.0,Nmr=4.0,Nadjustdist=1.0,ilength=5.0,flength=5.0,rdist=0.7,loopdiam=8.0,AHdist=1.41,AHLength=1.5,AHlength=1.41,ELdist=1.0}
\begin{picture}(100,118)(0,-118)
\unitlength=0.7mm
\node[NLangle=0.0](n0)(64.0,-52.0){B}

\node[NLangle=0.0](n1)(40.44,-72.0){C}

\node[NLangle=0.0,Nmarks=r](n3)(24.22,-92.0){E}

\node[NLangle=0.0,Nmr=0.0](n5)(52.22,-92.0){F}

\node[NLangle=0.0,Nmr=0.0](n6)(84.06,-72.28){D}

\node[NLangle=0.0,Nmarks=r](n7)(36.0,-116.0){I}

\node[NLangle=0.0,Nmarks=r](n8)(68.0,-116.0){J}

\node[NLangle=0.0,Nmarks=r](n9)(68.0,-92.0){G}

\node[NLangle=0.0,Nmarks=r](n10)(100.0,-92.0){H}

\node[NLangle=0.0,Nmarks=i,Nmr=0.0](n11)(36.0,-44.0){A}

\drawedge[ELside=r,ELdist=2.0](n0,n1){0}

\drawedge[ELside=r,ELdist=2.0](n1,n3){3}

\drawedge[ELdist=2.0](n1,n5){0}

\drawedge[ELside=r,ELdist=2.0](n6,n9){3}

\drawedge[ELdist=2.0](n6,n10){0}

\drawedge[ELdist=2.0](n5,n8){0}

\drawedge[ELside=r,ELdist=2.0](n5,n7){4}

\drawedge[ELdist=2.0,curvedepth=-8.0](n11,n1){0}

\drawedge[ELdist=2.0,ELside=r,curvedepth=8.0](n11,n0){0}






\drawedge[ELdist=2.0](n0,n6){0}

\end{picture}
  \caption{\label{fig:memory} Graph arena with a common weight function.}
  \vspace{-5mm}
\end{figure}

\begin{example}\label{ex:centipede}
  Consider the game arena $G$ of Fig. \ref{fig:memory}. We omit Player $2$'s weights since
  we are interested in computing the regret of Player $1$. Player 1's positions are circle
  nodes and Player 2's positions are square nodes. The target nodes are represented by double circles.
  The initial node is $A$. 
  Let $\lambda_1$ be the memoryless strategy defined by
  $\lambda_1(B) = C$ and $\lambda_1(C) = E$.
  For all $\lambda_2\in\Lambda_2(G)$, $\outcome^G_1(\lambda_1,\lambda_2)$ is either 
  $ACE$ or $ABCE$, depending on whether Player $2$ goes directly to $C$ or passes by $B$.
  In both cases, the outcome is winning and $\ut_1^G(\lambda_1,\lambda_2) = 3$. 
  What is the regret of playing $\lambda_1$ for Player $1$? To compute $\reg_G^1(\lambda_1)$, we
  should consider all possible strategies of Player $2$, but a simple observation allows us to
  restrict this range. Indeed, to maximize the regret
  of Player $1$, Player $2$ should cooperate in subtrees where $\lambda_1$ prevents to go, i.e.
  in the subtrees rooted at $D$ and $F$. Therefore we only have to consider
  the two following memoryless strategies $\lambda_2$ and $\lambda'_2$:
  both $\lambda_2$ and $\lambda'_2$ move from $F$ to $J$ and from $D$ to $H$, but
  $\lambda_2(A) = B$ while $\lambda'_2(A) = C$.
  In both cases, going to $F$ is a best response to $\lambda_2$ and $\lambda'_2$ for Player $1$,
  i.e. $\br_1^G(\lambda_2) = \br_1^G(\lambda'_2) = 0$. Therefore we get
  $\reg_1^G(\lambda_1,\lambda_2) = \ut_1^G(\lambda_1,\lambda_2) - \br_1^G(\lambda_2) = 3-0 = 3$.
  Similarly $\reg_1^G(\lambda_1,\lambda'_2) = 3$. Therefore $\reg_1^G(\lambda_1) = 3$.

  As a matter of fact, the strategy $\lambda_1$ minimizes the regret of Player $1$. Indeed, 
  if she chooses to go from $B$ to $D$, then Player $2$ moves from $A$ to $B$ 
  and from $D$ to $G$ (so that Player $1$ gets a utility $3$) and cooperates in the subtree
  rooted at $C$ by moving from $F$ to $J$. The regret of Player $1$ is therefore $3$.
    If Player $1$ moves from $B$ to $C$ and from $C$ to $F$, then Player $2$ moves
  from $A$ to $C$ and from $F$ to $I$ (so that Player $1$ gets a utility $4$), and from $D$ to $H$, 
  the regret of Player $1$ being therefore $4$. Similarly, one can show that all other strategies
  of Player $1$ have a regret at least $3$. Therefore $\reg_1^G = 3$.

  Note that the strategy $\lambda_1$ does not minimize the regret in the subgame 
  defined by the subtree rooted at $C$. Indeed, in this subtree, Player $1$ has
  to move from $C$ to $F$, and the regret of doing this is $4-3 = 1$. However the
  regret of $\lambda_1$ in the subtree is $3$. This example illustrates a situation where
  a strategy that minimizes the regret in the whole game does not necessarily minimize the
  regret in the subgames. Therefore we cannot apply a simple backward algorithm to compute
  the regret. As we will see in the next section, we first have to propagate some information
  in the subgames.
\end{example}




\section{Regret Minimization on Target-Weighted Graphs}\label{sec:targetgraph}

In this section, our aim is to give an algorithm to compute the regret
for Player $i$. This is done by reduction to a min-max game, defined
in the sequel. We say that we \textit{solve} the regret minimization
problem (\RMP for short) if we can compute the minimal regret and
a (finite representation of a) strategy that achieves this value.

\paragraph{Minmax games} Let $G = (S=S_1\uplus S_2, s_0,
T,\mu_1,\mu_2,\target_1,\target_2)$ be a $\twa$ and $i=1,2$.
We define the value $\minmax_i^G$ as follows:
$$
\minmax_i^G =
\min_{\lambda_i\in\Lambda_i(G)}\max_{\lambda_{\m i}\in\Lambda_{\m i}(G)} \ut_i^G(\lambda_i,\lambda_{\m i})
$$

\begin{proposition}\label{prop:minmaxgame}
  Given a \twa $G =
  (S,s_0,T,\mu_1,\mu_2,\target_1,\target_2)$, $i\in\{1,2\}$ and $K\in
  \mathbb{N}$, one can decide in time $O(|S|+|T|)$  whether 
  $\minmax_i^G\leq K$. The value $\minmax_i^G$ and a memoryless
  strategy that achieves this value can be computed in time 
  $O(log_2(M_i^G)(|S|+|T|))$. 
\end{proposition}

\begin{proof}

For all $j\geq 0$, we let $W_j\subseteq S$ be the set of positions
from which Player $i$ has a strategy to reach a position $s\in
\target_i$ in at most
$j$ steps, such that $\mu_i(s)\leq K$ and such that she does not pass by a position
$s'\in \target_i$ such that $\mu_i(s')>K$. Formally, we
denote by $\target_i^{>K}$ the set of positions $s\in \target_i$ 
s.t. $\mu_i(s)> K$. Then $W_0 = \target_i\backslash \target_i^{>K}$
 and for all $j>0$, $W_j = W_{j-1} \cup 
W_j^\exists\cup W_j^\forall$, where:
$$
\begin{array}{llllllll}
W_j^\exists & =&  \{ s\in S_i\backslash \target_i^{>K}\ |\ \exists s'\in W_{j-1},(s,s')\in T\} \\
W_j^\forall  & = &  \{ s\in S_{\m i}\backslash \target_i^{>K}\ |\ \forall
(s,s')\in T,\ s'\in W_{j-1}\}
\end{array}
$$

The sequence $W_0,W_1,\dots$ converges in at most $|S|$ steps to a set $W^*$, and
$\minmax_i^G\leq K$ iff $s_0\in W^*$. In order to compute
$W^*$ in time $O(|S|+|T|)$, we add counters to positions
that counts the number of their successors that are not already in the 
current set $W_j$. When adding a new node to $W_j$, we decrement the
counter of its predecessor by one (if it was not already $0$). Let $s$
be one of its predecessors and $c$ its counter value. If $s\in S_i$
and $c$ is strictly lesser than the number of its successors, $s$ will
be added to $W_{j+1}$. If $s\in S_{\m i}$ and $c=0$, then all the
successors of $s$ are in $W_j$, therefore $s$ will be added to $W_{j+1}$.
Now, in order to compute the value $\minmax_i^G$, we use the previous algorithm as the
building block of a dichotomy algorithm that starts with the maximal
finite value which can be achieve by Player $i$ if she has a winning strategy
to its target, i.e. $M^G_i$.

\noindent If $\minmax_i^G = +\infty$, then any strategy achieves this
value. Otherwise in order to extract a strategy, it suffices to keep
for each position $s\in W_j\cap S_i$, a pointer to a position
$s'\in W_{j-1}$ such that $(s,s')\in T$ when computing
the sequence of $W_j$'s. Note that this strategy is memoryless. 
\end{proof}

Since roles of the players are symmetric, without loss of generality we can 
focus on computing the regret of Player $1$ only.
Therefore we do not consider Player $2$'s targets and weights.
Let  $G = (S= S_1\uplus S_2,
s_0,T,\mu_1,\target_1)$ be a \twa (assumed to be fixed from now on). 
Let $\lambda_1\in\Lambda_1(G)$ be a winning
strategy of Player $1$ (if it exists). Player $2$ can enforce
Player $1$ to follow one of the paths of
$\outcome^{G,\target_1}(\lambda_1)$ by choosing a suitable strategy. When choosing a path
$\pi\in\outcome^{G,\target_i}(\lambda_1)$, in order to maximize the
regret of Player $1$, Player $2$ cooperates (i.e. she minimizes the utility) if
Player $1$ would have deviated from $\pi$. This leads to the notion of
\textit{best alternative} along a path. Informally, the best
alternative along $\pi$ is the minimal value Player $1$ could have achieved
if she deviated from $\pi$, assuming Player $2$ cooperates. 
Since Player $2$ can enforce one of the paths of
$\outcome^{G,\target_1}(\lambda_1)$, to maximize the regret of Player
$1$, she will choose the path $\pi$ with the
highest difference between $\ut_1^G(\pi)$ and the minimal best
alternative along $\pi$. As an example consider the \twa arena of Fig. \ref{fig:memory}.
In this example, if Player $1$ moves from
$C$ to $E$, then along the path $ACE$, the best alternative is $0$. Indeed, the other alternative
was to go from $C$ to $F$ and in this case, Player $2$ would have
cooperated.

We now formally define the notion of best alternative. 
Let $s\in S_1$. The best value that can be achieved from
$s$ by Player $1$ when Player $2$ cooperates is defined by:

$$
\best^G_1(s) = \min_{\lambda_1\in \Lambda_1(G,s)}
\min_{\lambda_2\in\Lambda_2(G,s)} \ut^{(G,s)}_1(\lambda_1,\lambda_2)
$$

Let $(s,s')\in T$. The \textit{best
  alternative} of choosing $s'$ from $s$ for Player $1$, denoted by $\ba^G_1(s,s')$, 
is defined as the minimal value she could have achieved by choosing 
another successor of $s$ (assuming Player $2$ cooperates). Formally:


$$
\ba^G_1(s,s') = \left\{\begin{array}{lllll}
+\infty & \text{ if } s\in S_{2} \\
\min_{(s,s'')\in T,s''\neq s'} \best^G_1(s'') & \text{ if } s\in S_{1} \\
\end{array}\right.
$$

with $\min\varnothing = +\infty$. Finally, the best
alternative of a path $\pi = s_0s_1\dots s_n$ is
defined as $+\infty$ if $n=0$ and as the minimal
best alternative of the edges of $\pi$ otherwise:

$$
\ba_G^1(\pi) = \min_{0\leq j<n} \ba^1_G(s_j,s_{j+1})
$$

We first transform the graph $G$ into a
graph $G'$ such that all the paths that lead to a node $s$ have the same
best alternative. This can be done since the number of best
alternatives is bounded by $|\target_1|$. The construction of $G'$ is
 done inductively by storing the best alternatives in the positions.

\begin{definition}\label{def:reduc1}
The graph of best alternatives of $G$ is the $\twa$
$G' = (S'=S'_1\uplus S'_2 ,s_0',T',\mu'_1,\target'_1)$ defined by: 

\begin{list}{\labelitemi}{\leftmargin=0.5em}
\item $S'_i = S_i \times ([M_1^G]\cup \{+\infty\})$, $i=1,2$ and $s_0' = (s_0,+\infty)$;
\item for all $(s,b_1),(s',b'_1)\in S'$, $((s,b_1),(s',b'_1))\in T'$
  iff $(s,s')\in T$ and 
  \vspace{-3mm}
  $$
  b'_1 = \left\{\begin{array}{llll} \min (b_1,\ba_1^G(s,s')) &\text{ if } s\in S_1 \\
        b_1 & \text{ if } s\in S_2\end{array}\right.
  $$
  \vspace{-3mm}

\item $\target_1' = S'_1\cap (C_1\!\times\! [M_1^G])$ and
$\forall (s,b)\!\in\! \target'_1,\mu'_1(s,b) = \mu_1(s)$.
\end{list}
\end{definition}

\begin{proposition}\label{prop:ba}
  For all $(s,b)\in S'$ and all finite path $\pi$ in $G'$ from $(s_0,+\infty)$
  to $(s,b)$, $\ba_1^{G'}(\pi) = b$. 
\end{proposition}

Because the number of best alternatives is bounded by $|\target_1|$,
the game $G'$ can be constructed in polynomial time:

\begin{proposition}\label{prop:construction_G'}
  $G'$ can be constructed in time
  $O\left((|\target_1|+\log_2(M_1^G))\times(|S|+|T|)\right)$.
\end{proposition}

Since the best alternative information depends only on the paths, the paths of $G$ and those of $G'$
are in bijection. This bijection can be extended to strategies. In
particular, we define two mappings $\Phi_i$ from $\Lambda_i(G)$ to
$\Lambda_i(G')$, for all $i=1,2$. For all path $\pi = s_0s_1\dots$ in $G$ (finite
  or infinite), we denote by $B(\pi)$ the path of $G'$ defined by
  $(s_0,b_0) (s_1,b_1)\dots$ where $b_0= +\infty$ and for all $j>0$,
  $b_j = \ba_G^1(s_0\dots s_{j-1})$. The mapping $B$ is bijective, and
  its inverse corresponds to $\proj_1$.

  The mapping $\Phi_i$ maps any strategy
  $\lambda_i\in\Lambda_i(G)$ to a strategy $\Phi_i(\lambda_i)\in
  \Lambda_i(G')$ such that $\Phi_i(\lambda_i)$ behaves as $\lambda_i$
  on the first projection of the play and adds the best alternative
  information to the position. Let $h\in S'^*$ such that
  $\last(h)\in S_i'$. Let $s = \lambda_i(\proj_1(h))$. Then 
  $\Phi_i(\lambda_i)(h) = (s, \ba_G^1(\proj_1(h).s))$. 
  The inverse mapping $\Phi_i^{-1}$ just projects the best alternative
  information away. In particular, for all
  $\lambda'_i\in\Lambda_i(G')$, and all $h\in S^*$ such that
  $\last(h)\in S_i$, $\Phi_i^{-1}(\lambda_i)(h) =
  \proj_1(\lambda_i(B(h)))$.

\noindent Then, $\Phi_i$'s are bijective and $\Phi_1$ preserves the regret values:

\begin{lemma}\label{lem:bij}
  \mbox{$\forall \lambda_1\!\in\!\Lambda_1(G),\reg_1^G(\lambda_1)\!=
  \!\reg^{G'}_1(\Phi_1(\lambda_1))$}.
\end{lemma}


The best alternative information is crucial to compute the regret. This is a global information
that allows us to compute the regret locally, as stated by the next lemma. For all $(s,b)\in \target'_1$, we let 
$\nu_1(s,b) = \mu_1(s) - \min (\mu_1(s), b)$. We extend
$\nu_1$ to pairs of strategies as usual  -- $\nu_1(\lambda_1,\lambda_2)$
being infinite if $\lambda_1$ is losing --.

\begin{lemma}\label{lem:maxmax}
  \mbox{$\forall \lambda_1\in\Lambda_1(G')$,
  $\reg^{G'}_1(\lambda_1)\! =\!\!\!\!\! \max\limits_{\lambda_2\in\Lambda_2(G')}
  \nu_1(\lambda_1,\lambda_2)$}.
\end{lemma}

\begin{proof}(Sketch)
    It is clear if $\lambda_1$ is losing. If it is winning, then
    let $\lambda_2$ which maximizes $\reg^{G'}_1(\lambda_1)$ and 
    $\pi = \outcome^{G',\target_1}(\lambda_1,\lambda_2)$. 
    Without changing the regret values, we can assume that $\lambda_2$
    cooperates if Player $1$ would have deviated from $\pi$, i.e.
    $\lambda_2$ minimizes the utility in the subgames $(G,s)$ where
    $s$ is not the successor of some element of $\pi$. The best response
    to $\lambda_2$ is either the value $\ut^{G'}_1(\lambda_1,\lambda_2)$, i.e.
    $\mu_1(\last(\pi))$, or the minimal best alternative along $\pi$.
    By Proposition \ref{prop:ba}, this minimal best alternative along $\pi$ is
    exactly $\proj_2(\last(\pi))$. Therefore $\br^{G'}_1(\lambda_2) = \min(\mu'_1(\last(\pi)),\ba_{G'}^1(\pi))$
    and $\reg^{G'}_1(\lambda_1) = \nu^1(\last(\pi)) = \nu^1(\lambda_1,\lambda_2)$.
    Conversely, for any strategy $\lambda_2$ which maximizes $\nu^1(\lambda_1,\lambda_2)$,
    we can also assume without changing the value $\nu^1(\lambda_1,\lambda_2)$ that
    $\lambda_2$ cooperates if Player $1$ would have deviated
    from $\outcome^{G'}(\lambda_1,\lambda_2)$, and we therefore have
    $\reg^{G'}_1(\lambda_1,\lambda_2) = \nu^1(\lambda_1,\lambda_2)$.
\end{proof}

We can now reduce the $\RMP$ to a min-max problem~:

\begin{lemma}\label{lemma:regretminmax}
  Let $H = (S',s'_0,T',\nu^1,\target_1')$ where $S',s'_0,T',\target_1'$ are defined in Definition \ref{def:reduc1}. 
  Then
  $$
  \reg_1^G = \minmax^{H}_1
  $$


\end{lemma}

\noindent \textit{Proof}
It is a consequence of Lemmas \ref{lem:bij} and \ref{lem:maxmax}:

$$
\begin{array}{lllllll}
\reg_1^G  & = & \min\limits_{\lambda_1\in\Lambda_1(G)}

\reg_1^G(\lambda_1) & \text{(definition)} \\

& = & \min\limits_{\lambda_1\in\Lambda_1(G)} 
\reg^{H}_1(\Phi_1(\lambda_1)) & \text{(Lemma \ref{lem:bij})}\\ 

& = & \min\limits_{\lambda_1\in\Lambda_1(H)}
\reg^{H}_1(\lambda_1) & \text{(Lemma \ref{lem:bij})} \\

& = & \min\limits_{\lambda_1\in\Lambda_1(H)}
\max\limits_{\lambda_2\in\Lambda_2(H)} \nu^1(\lambda_1,\lambda_2) &
\text{(Lemma \ref{lem:maxmax})} \hfill\square  \\ 
\end{array}
$$

As a consequence of Propositions~\ref{prop:minmaxgame},
\ref{prop:construction_G'} and Lemma~\ref{lemma:regretminmax}, we can
solve the \RMP on \twa's. We first compute the graph of best
alternatives and solve a $\minmax$ game. This gives us a memoryless
strategy that achieves the minimal regret in the graph of best
alternatives. To compute a strategy in the original graph,
we apply the inverse mapping $\Phi_1^{-1}$: this gives a finite-memory strategy 
whose memory is exactly the best alternative seen
along the current finite play. Therefore the needed memory is bounded by
the number of best alternatives, which is bounded by $|C_1|$.

\begin{theorem}\label{thm:rmp:twa}
  The \RMP on a \twa $G = (S,s_0,T,\mu_1,\target_1)$ can be solved in
  $O\left(|\target_1|\cdot log_2(M_1^G)\cdot(|S|+|T|)\right)$.
\end{theorem}


\section{Regret Minimization in Edge-Weighted Graphs}\label{sec:graphs}

In this section, we give a pseudo-polynomial time algorithm to solve
the \RMP in weighted arenas (with weights on edges). In a first step,
we prove that if the regret is finite, the strategies minimizing the regret generates outcomes
whose utility is bounded by some value which depends on the graph.
This allows us to reduce the problem to the \RMP in a \twa, which can
then be solved by the algorithm of the previous section.

Let $G = (S = S_1\uplus S_2,s_0,T,\mu_1,\target_1)$ be a weigthed game arena with
objective $\target_1$. As in the previous section, we assume that we want to minimize the regret
of Player 1, so we omit the weight function and the target of Player $2$.

\begin{definition}[Bounded strategies]
Let $B\in\mathbb{N}$ and $\lambda_1\in\Lambda_1(G)$. The strategy $\lambda_1$ is bounded by $B$ if for all
$\lambda_{2}\in \Lambda_{2}(G)$, $\ut^G_1(\lambda_1,\lambda_{2})\leq B$.
\end{definition}

Note that a bounded strategy is necessarily winning, since by definition, the utility of some outcome 
is infinite iff it is loosing. 
The following lemma states that the winning strategies that minimize the regret of Player $1$ are bounded.

\begin{lemma}\label{lem:regretbounded}
    For all weighted arena $G = (S,s_0,T,\mu_1,\target_1)$ and for any strategy $\lambda_1\in\Lambda_1(G)$ winning in $G$ for Player $1$ that minimizes her regret,
    $\lambda_1$ is bounded by $2 M^G |S|$.
\end{lemma}

\begin{proof}
    Since we consider reachability games, it is well-known that if there is a winning strategy for Player $1$, there is a memoryless strategy
    $\gamma_1$ winning for Player $1$ (see for instance \cite{GraedelTW02}). In particular, for all $\lambda_{2}\in\Lambda_{2}(G)$, 
    $\outcome^{G,\target_1}(\gamma_1,\lambda_{2})$ does not contain twice the same position. Indeed, if there is a loop, since the strategy is memoryless, 
    Player $2$ can enforce Player $1$ to take this loop
    infinitely many times. Therefore for all $\lambda_{2}\in \Lambda_{2}(G)$, 
    $\ut^G_1(\gamma_1,\lambda_{2})\leq M^G |S|$. Therefore the
    following holds: $(\star)\ \forall \lambda_{2}\in \Lambda_{2}(G), \br_1^G(\lambda_{2})\leq M^G |S|$.
    Moreover, $\reg_1^G(\gamma_1)\leq M^G |S|$. Indeed, let $\lambda_{2}$ which
    maximizes $\reg_1^G(\gamma_1,\lambda_{2})$. Then $\reg_1^G(\gamma_1) = \ut^G_1(\gamma_1,\lambda_{2}) - \br_1^G(\lambda_{2})$.
    Since $\ut_1^G(\gamma_1,\lambda_{2})\leq M^G |S|$ and $0\leq \br_1^G(\lambda_{2}) \leq M^G |S|$, we get
    $\reg_1^G(\gamma_1)\leq M^G |S|$.  Thus $(\star\star)\ \reg_1^G\leq M^G |S|$.

    \noindent Finally let $\lambda_1$ be a winning strategy which minimizes the regret of Player $1$, and $\lambda_{2}\in\Lambda_{2}(G)$.
    We have $\reg_1^G(\lambda_1,\lambda_{2})\leq M^G |S|$ (by
    $(\star\star)$), therefore
    $\ut_1^G(\lambda_1,\lambda_{2}) - \br_1^G(\lambda_{2})\leq M^G |S|$, which gives
    $\ut_1^G(\lambda_1,\lambda_{2})\leq 2 M^G |S|$ (by $(\star)$).
\end{proof}

Let $B = 2M^G |S|$. 
Thanks to Lemma \ref{lem:regretbounded} we can reduce the $\RMP$ in a weighted
arena into the $\RMP$ in a \twa. Indeed, it suffices to enrich every position
of the arena with the sum of the weights occuring along the path used to reach this position.
A position may be reachable by several paths, therefore it will be duplicated as many times as they are
different path utilities. This may be unbounded, but Lemma \ref{lem:regretbounded} ensures that it is sufficient to sum the weights 
up to $B$ only. This may results in a larger graph, but
its size is still pseudo-polynomial (polynomial in the maximal weight
and the size of the graph). 

\begin{definition}\label{def:reduc2}
Let $G = (S=S_1\uplus S_2,s_0,T,\mu_1,\target_1)$ be a weigthed game arena. The graph of utility
is the \twa  $G' = (S'=S'_1\uplus S'_2,s'_0,T',\mu'_1,\target_1')$ defined by: 

\begin{list}{\labelitemi}{\leftmargin=1em}
\item $S'_i \subseteq S_i\times [B]$, $i=1,2$ and $s'_0 = (s_0,0)$;
\item for all $(s,u),(s',u')\in S'$, $((s,u),(s',u'))\in T'$ iff
  $(s,s')\in T$ and $u' =  u+\mu_1(s,s')$ ;
\item $\target'_1 = (\target_1\times [B])\cap S'$ and $\forall (s,u)\in \target'_1$,
$\mu'_1(s,u) = u$.
\end{list}
\end{definition}

We now prove that $\reg_1^G = \reg^{G'}_1$. 
The utility information added to the nodes of $G$ is uniquely determined 
by the path used to reach the current position. Therefore the strategies
of both players in $G$ can naturally be mapped to strategies in $G'$.
More formally, we define a mapping $\Phi$ from $\Lambda_1(G)\cup \Lambda_2(G)$ into
$\Lambda_1(G')\cup \Lambda_2(G')$. Let $i\in\{1,2\}$ and $\lambda_i\in\Lambda_i(G)$. 
Let $h\in\play_f(G')$ such that $\last(h)\in S'_i$. Let $s = \lambda_i(\proj_1(h))$ and $u = \mu_1(\proj_1(h).s)$.
Then

$$
\Phi(\lambda_i)(h) = 
\left\{\begin{array}{llllllll}
\perp & \text{if } u>B \\
(s,u) & \text{otherwise}
\end{array}\right.
$$

The mapping $\Phi$ is surjective, but not necessarily injective. Indeed, two strategies that behave
similarly up to an utility $B$ are mapped to the same strategy in $G'$.
Let $\lambda'_i\in\Lambda_i(G')$. Any strategy $\lambda_i\in
\Lambda_i(G)$ that behaves like $\lambda'_i$ (on the first
projections of plays) while the utility of the play is
bounded by $B$ is a preimage of $\lambda'_i$. Formally, 
for all $h = s_0s_1\dots s_n \in\play(G)$, we let $\tilde{h} =
(s_0,u_0) (s_1,u_1)\dots (s_n,u_n)$ where for all $j$, $u_i =
\mu_1(s_0\dots s_j)$. Then, any strategy
$\lambda_i\in\Lambda_i(G)$
is a preimage of $\lambda'_i$ iff for all finite play
$h\in\play(G)$ such that $\last(h)\in S_i$, all $s\in S$, 
all $u\in\mathbb{N}$, if $\lambda'_i(\tilde{h})$ is defined
and equal to $(s,u)$, then $\lambda_i(h) = s$. 

\begin{lemma}\label{lem:surjection}
  For all $i=1,2$, $\Phi(\Lambda_i(G)) = \Lambda_i(G')$.
\end{lemma}

\noindent We denote by $\Lambda_1^{\leq B}(G)$ the set of strategies bounded by
$B$. The mapping $\Phi$ preserves the regret values for bounded strategies:

\begin{lemma}\label{lem:preserveregret}
  $\forall \lambda_1\in\Lambda_1^{\leq B}(G)$, $\reg_1^G(\lambda_1) =
  \reg^{G'}_1(\Phi(\lambda_1))$.
\end{lemma}

\begin{proof}(Sketch) This lemma is supported by the following result: for all $\lambda_1\in\Lambda_1(G)$ and
  $\lambda_2\in\Lambda_2(G)$. If $\outcome^{G'}(\Phi(\lambda_1),\Phi(\lambda_2))$ is winning for
  Player $1$ in $G'$ or $\ut_{1}^G(\lambda_1,\lambda_2)\leq B$, then $\ut^G_1(\lambda_1,\lambda_2) =
  \ut^{G'}_1(\Phi(\lambda_1),\Phi(\lambda_2))$. 
\end{proof}

Note that any strategy $\lambda_1\in \Lambda_1(G)$ is bounded by $B$ iff 
$\Phi(\lambda_1)$ is winning in $G'$ for Player $1$. We can now prove
the correctness of the reduction:

\begin{lemma}\label{lemma:sameregrets_G_G':graphs}
    $\reg^G_1  = \reg^{G'}_1$
\end{lemma}

\begin{proof}
    Suppose that $\reg^G_1 = +\infty$. If $\reg^{G'}_1<+\infty$, then
    there is a strategy $\lambda'_1\in \Lambda_1(G')$ winning in $G'$ for Player $1$. 
    By Lemma \ref{lem:surjection}, $\lambda'_1 = \Phi(\lambda_1)$ for some $\lambda_1\in \Lambda_1(G)$.
    Since $\Phi(\lambda_1)$ is winning, $\lambda_1$ is bounded by $B$, and \textit{a fortiori} winning. Thus $\reg^G_1 < +\infty$,
    which is a contradiction. Therefore $\reg^{G'}_1 = \reg^G_1 = +\infty$.

    \noindent Suppose that $\reg^G_1 < +\infty$. Thus there is a winning strategy $\lambda_1$ which minimizes the regret.
    By Lemma \ref{lem:regretbounded}, $\lambda_1$ is bounded $B$. By Lemma \ref{lem:preserveregret}, $\reg^G_1(\lambda_1) = \reg^{G'}_1(\Phi(\lambda_1))$.
    Thus $\reg^G_1 = \reg^G_1(\lambda_1) = \reg^{G'}_1(\Phi(\lambda_1)) \geq \reg^{G'}_1$. Conversely, since
    $\Phi(\lambda_1)$ is winning in $G'$, there is a winning strategy $\gamma'_1\in \Lambda_1(G')$ minimizing the regret.
    By Lemma \ref{lem:surjection}, $\gamma'_1 = \Phi(\gamma_1)$ for some $\gamma_1\in \Lambda_1(G)$. Since $\Phi(\gamma_1)$ is winning, 
    $\gamma_1$ is bounded by $B$, and by Lemma \ref{lem:preserveregret}, $\reg^G_1(\gamma_1) = \reg^{G'}_1(\gamma'_1)$.
    So $\reg_1^G \leq \reg_1^G(\gamma_1) = \reg^{G'}_1(\gamma'_1) = \reg^{G'}_1$.
\end{proof}

To solve the \RMP for a weighted arena $G$, we first construct the graph of utility $G'$, and
then apply Theorem \ref{thm:rmp:twa}, since $G'$ is a \twa. Correctness is ensured
by Lemma \ref{lemma:sameregrets_G_G':graphs}. This returns a finite-memory strategy of $G'$ that
minimizes the regret, whose memory is the best alternative seen so far.
To obtain a strategy of $G$ minimizing the regret, one applies the inverse
mapping $\Phi^{-1}$ defined previously. This gives us a finite-memory strategy whose
memory is the utility of the current play up to $M^G$ and the best alternative
seen so far.

\begin{theorem}\label{thm:rmp:graphs}
  The \RMP on a weighted arena $G = (S=S_1\uplus S_2,s_0,T,\mu_1,\target_1)$
  can be solved in time $O\left((M^G)^2\cdot log_2(|S|\cdot M^G)\cdot |S|\cdot
  \target_1\cdot(|S|+|T|)\right)$.
\end{theorem}

\section{Iterated Regret Minimization (IRM)}\label{sec:iterate}

In this section, we show how to compute the iterated
regret for tree arenas and for weighted arenas where weights are strictly
positive (by reduction to a tree arena).

Let $G = (S=S_1\uplus S_2,s_0,T,\mu_1,\mu_2, \target_1, \target_2)$ 
be a weighted arena. Let $i\in\{1,2\}$, $P_i\subseteq \Lambda_i(G)$
and $P_{\m i}\subseteq \Lambda_{\m i}(G)$. The regret of Player $i$ when she plays 
strategies of $P_i$ and when Player $\m i$ plays strategies of $P_{\m i}$ is defined by:

$$
\begin{array}{lllllll}
\reg^{G,P_i,P_{\m i}}_i & = & \min\limits_{\lambda_i\in P_i}\max\limits_{\lambda_{\m i}\in P_{\m i}} \ut_i^G(\lambda_i,\lambda_{\m i}) - \br^{G,P_i}_i(\lambda_{\m i})\\
\br^{G,P_i}_i(\lambda_{\m i}) & = & \min_{\lambda^*_i\in P_i} \ut_{i}^G(\lambda^*_i,\lambda_{\m i})
\end{array}
$$

For all $\lambda_i\in P_i$ and $\lambda_{\m i}\in P_{\m i}$, we define $\reg_i^{G,P_i,P_{\m i}}(\lambda_i)$ and
$\reg_i^{G,P_i,P_{\m i}}(\lambda_i,\lambda_{\m i})$ accordingly. We
now define the strategies of rank $j$, which are the one that survived
$j$ times the deletion of strictly dominated strategies.
The strategies of rank $0$ for Player $i$ is $\Lambda_i(G)$.
The strategies of rank $1$ for both players are 
those which minimize their regret against strategy of rank $0$.  More generally,
the strategies of rank $j$ for Player $i$ are the strategies of rank
$j-1$  which minimize her regret against Player $\m i$'s strategies of
rank $j-1$. Formally, strategies of rank $j$ are obtained via a
\textit{delete} operator $D : 2^{\Lambda_1(G)}\times
2^{\Lambda_2(G)}\rightarrow 2^{\Lambda_1(G)}\times
2^{\Lambda_2(G)}$ such that for all $P_1\subseteq \Lambda_1(G)$ and
all $P_2\subseteq \Lambda_2(G)$, 

$$
\begin{array}{lcc}
  & & \{ \lambda_1\in P_1 | \reg_1^{G,P_1,P_2} = \reg_1^{G,P_1,P_2}(\lambda_1) \} \\
D(P_1,P_2) & = &  \times  \\
& &  \{ \lambda_2\in P_2 | \reg_2^{G,P_2,P_1} = \reg_2^{G,P_2,P_1}(\lambda_2) \}
\end{array}
$$

\noindent We denote by $D^j$ the composition of $D$ $j$ times. 



\begin{definition}[$j$-th regret]
    Let $j\geq 0$. The set of strategies of rank $j$ for Player $i$ is $P^j_i = \proj_i(D^j(\Lambda_1(G),\Lambda_2(G)))$.
    The $j+1$-th regret for Player $i$ is defined by $\reg^{G,j+1}_{i}
    = \reg^{G,P^j_i,P_{\m i}^j}_i$. In particular, $\reg^{G,1}_{i} =
    \reg^G_i$.
\end{definition}

\begin{proposition}\label{lem:decreaseregret}
  Let $i\in\{1,2\}$. For all $j\geq 0$, $P^{j+1}_i\subseteq P^j_i$ and 
  $\reg^{G,j+1}_i\leq \reg^{G,j}_i$.
\end{proposition}

\begin{proof}(Sketch)
  $P^{j+1}_i\subseteq P^j_i$ is by definition the operator $D$.
  For all $\lambda_{\m i}\in P^j_{\m i}$,
  $\br^{G,P^{j+1}_i}(\lambda_{\m i})\geq \br^{G,P^j_i}(\lambda_{\m i})$ (because
  we minimize over less strategies).
  Thus for all $\lambda_i\in P^j_i$ and $\lambda_{\m i}\in
  P^j_{\m i}$, $\reg^{G,j+1}_i(\lambda_i,\lambda_{\m i}) \leq
  \reg^{G,j}_i(\lambda_i,\lambda_{\m i})$. 
  Since $P^{j+1}_{\m i}\subseteq P^j_{\m i}$, 
  $\reg^{G,j+1}_i(\lambda_i) \leq \reg^{G,j}_i(\lambda_i)$ (because we
  maximize over less strategies).  Therefore $\reg^{G,j+1}_i\leq
  \reg^{G,j+1}_i(\lambda_i) \leq \reg^{G,j}_i(\lambda_i) =
  \reg^{G,j}_i$.
\end{proof}

Clearly, the sequence of regrets converges:

\begin{proposition}\label{lem:stable}
  There is an integer $\star \geq 1$ such that for all $j\geq \star$,
  for all $i\in\{1,2\}$, $\reg^{G,j}_i = \reg^{G,\star}_i$.
\end{proposition}


\begin{definition}[iterated regret]
  For all $i=1,2$, the iterated regret of Player $i$ is
  $\reg^{G,\star}_i$.
\end{definition}

\begin{example}
  As we already saw in the Centipede Game depicted on
  Fig.~\ref{fig:centipede}, the Player $1$'s strategy minimizing
  her regret is to stop at the last step (move from $A$ to
  $B$, from $C$ to $D$ and from $E$ to $S$). Its regret value is $1$.
  The Player $2$'s strategy minimizing her regret is also to stop
  at the last step, i.e. to move from $B$ to $C$ and from $D$ to $E$,
  her regret being $1$. Therefore $\reg^G_1 = \reg^{G,1}_1 = 1$ and
  $\reg^G_2 = \reg^{G,1}_2 = 1$.  If Player $1$ knows that Player $2$
  will ultimately move to $E$, she can play the same strategy as
  before, and her regret $\reg^{G,2}_1$ is $0$. Similarly
  $\reg^{G,2}_2 = 0$.  Therefore $\reg^{G,\star}_1 = \reg^{G,\star}_2
  = 0$.
\end{example}

\subsection{IRM in Tree Arenas}\label{sec:iterate:trees}

In this section, we let $i\in\{1,2\}$ and $G = (S=S_1\uplus
S_2,s_0,T,\mu_1,\mu_2,\target_1,\target_2)$ be a finite
edge-weighted tree arena. We can transform $G$ into a target-weighted tree arena such that
$\target_1 = \target_2$ (denoted by $\target$ in the sequel) is the set of leaves of the tree, if we allow
the functions $\mu_i$ to take the value $+\infty$. This transformation
results in a new target-weighted tree arena $G' = (S=S_1\uplus
S_2,s_0,T,\mu'_1,\mu'_2,\target)$ with the same set of states
and transitions as $G$ and for all leaf $s\in \target$, 
$\mu'_i(s) = \ut^{G'}_i(\pi)$, where $\pi$ is the root-to-leave
path leading to $s$. The time
complexity of this transformation is $O(|S|)$.


We now assume that $G = (S=S_1\uplus
S_2,s_0,T,\mu_1,\mu_2,\target)$ is a target-weighted tree arena where
$\target$ is the set of leaves. Our goal is to define a delete
operator $D$ such that $D(G)$ is a subtree of $G$ such that
for all $i=1,2$, $\Lambda_i(D(G))$ are the strategies of $\Lambda_i(G)$ that minimize
$\reg^G_i$. In other words, any pairs of subsets of strategies for
both players in $G$ can be represented by a subtree of $G$. This is
possible since all the strategies in a tree
arena are memoryless.
A set of strategies $P_i\subseteq \Lambda_i(G)$ is therefore represented by removing from $G$ all the edges
$(s,s')$ such that there is no strategy $\lambda_i\in P_i$ such that $\lambda_i(s) = s'$.
In our case, one first compute the sets of strategies that minimize
regret. This is done as in Section \ref{sec:targetgraph} by constructing
the tree of best alternatives $H$ (but in this case with the best alternative of both players) and by solving
a min-max game. From $H$ we delete all edges 
that are not compatible with a strategy that minimize the $\minmax$ value of some player.
We obtain therefore a subtree $D(H)$ of $H$ such that any strategy of $H$ 
is a strategy of $D(H)$ for Player $i$
iff it minimizes the $\minmax$ value in $H$ for Player $i$. By projecting away the best alternative information
in $D(H)$, we obtain a subtree $D(G)$ of $G$ such that any Player $i$'s strategy of $G$ is a strategy of $D(G)$
iff it minimizes Player $i$'s regret in $G$. We can iterate this process to compute the iterated regret, and
we finally obtain a subtree $D^*(G)$ such that any strategy of $G$ minimizes the iterated regret
for Player $i$ iff it is a Player $i$'s strategy in $D^*(G)$.


\begin{definition}\label{def:reduc1:tree}
  The tree of best alternatives of $G$ is the tree $H = (S'=S'_1\uplus
  S'_2 ,s_0',T',\mu'_1,\mu'_2,\target')$ defined by:
\vspace{-1mm}
\begin{list}{\labelitemi}{\leftmargin=0.5em}
\vspace{-1mm}
\item $S'_i = \{ (s,b_1,b_2) \; | \; s\in S_i, b_\kappa= \ba^G_\kappa(\pi_s),\kappa=1,2\}$, where $\pi_s$ is the path from the root $s_0$ to $s$;
\vspace{-1mm}
\item $s_0' = (s_0,+\infty,,+\infty)$;
\vspace{-1mm}
\item $\forall s,s'\in S'$, $(s,s')\in T'$ iff $(\proj_1(s),\proj_1(s'))\in T$
\vspace{-1mm}
\item $\target' = \{s\in S'\ |\ \proj_1(s)\in\target\}$;
\vspace{-1mm}
\item \mbox{$\forall (s,b_1,b_2)\!\in\! \target',\mu_i'(s,b_1,b_2) = \mu_i(s) - \min (\mu_i(s),b_i)$}.
\end{list}
\end{definition}

Note that $H$ is isomorphic to $G$. There is indeed a one-to-one
mapping $\Phi$ between the states of $G$ and the states of $H$: for
all $s\in S$, $\Phi(s)$ is the only state $s'\in S'$ of the form
$s'=(s,b_1,b_2)$. Moreover, this mapping is naturally extended to
strategies. Since all strategies are memoryless, any strategy
$\lambda_i\in\Lambda_i(G)$ is a function $S_i \to S$. Thus, for all
$s'\in S'_i$, $\Phi(\lambda_i)(s')=
\Phi\left(\lambda_i(\Phi^{-1}(s'))\right)$. Without loss of generality
and for a technical reason, we assume that any strategy $\lambda_i$ is
only defined for states $s\in S_i$ that are compatible with this
strategy, i.e. if $s$ is not reachable under $\lambda_i$ then the
value of $\lambda_i$ does not need to be defined. The lemmas of
Section~\ref{sec:targetgraph} still hold for the tree $H$:

\begin{lemma}\label{lem:mapstwo:tree}
For all $i\in\{1,2\}$, $\Phi(\Lambda_i(G)) = \Lambda_i(H)$ and any strategy $\lambda_i\in \Lambda_i(G)$ minimizes
$\reg_i^G$ iff $\Phi(\lambda_i)$ minimizes $\minmax_i^{H}$. Moreover $\reg_i^G = \minmax_i^{H}$.
\end{lemma}

As in Section~\ref{sec:targetgraph}, the \RMP on a tree arena can 
be solved by  min-max game. For all $s\in S'$, 
we define $\minmax_i^{H}(s) = \minmax_i^{(H,s)}$ and compute this value
by a backward induction algorithm. In particular, $\minmax_i^{H} = \minmax_i^{H}(s'_0)$ and for all $s\in S'$:

\vspace{-2mm}
$$
\minmax_i^{H}(s)= \left\{\begin{array}{llllll}
\mu_i'(s) & \!\!\!\text{if $s\in\target'$} \\
\min_{(s,s')\in T'} \minmax_i^{H}(s') & \!\!\!\text{if } s\in S'_i \\
\max_{(s,s')\in T'} \minmax_i^{H}(s') & \!\!\!\text{if } s\in S'_{\m i} \\
\end{array}\right.
$$
\vspace{-4mm}
 

\begin{theorem}\label{thm:rmp:tree}
  The \RMP on a tree arena $G = (S,s_0,T,\mu_1,\mu_2,\target)$ can be solved in
  $O\left(|S|\right)$.
\end{theorem}

The backward algorithm not only allows us to compute $\minmax_i^{H}$
for all $i\in\{1,2\}$, but also to compute a subtree $D(H)$ that
represents all the Player $i$'s strategies that achieve this value. We
actually define the operator $D$ in two steps. First, we remove the
edges $(s,s')\in T'$, such that $s\in S'_i$ and $\minmax_i^{H}(s') >
\minmax_i^{H}$ for all $i=1,2$. We obtain a new graph $H'$ consisting
of several disconnected tree components. In particular, 
there are some states no longer reachable from the root $s_0'$. Then
we keep the connected component that contains $s'_0$ and obtain
a new tree $D(H)$.

Player $i$'s strategies in $D(H)$ are not in the stricter sense strategies of
$H$, as they do not specify what to play when Player $\m i$ leads
Player $i$ to a position that is not in $D(H)$. 
More formally, let $\lambda_i$ be a strategy of Player $i$ defined on
$D(H)$ and $\lambda_{\m i}$ a strategy of Player $\m i$ on $H$. 
If there is a position $s$ of $D(H)$ owned by Player $\m i$ such that
$\lambda_{\m i}$ leads to $s$ when Player $i$ plays $\lambda_i$, and
if $\lambda_{\m i}(s) = s'$ for some position $s'$ not in $D(H)$, 
then $\lambda_i(s')$ is undefined. This never happens when $\lambda_i$ is
opposed to a strategy $\lambda_{\m i}$ of $D(H)$, but may happen when
opposed to a strategy $\lambda_{\m i}$ of $H$. For this reason, we
define the strategies $\lambda_i$ of $D(H)$ for Player $i$ as
the strategies of $H$ such that for all $s\in S'_i$, $(s,\lambda_i(s))$ 
is an edge of $H'$. We denote again by $\Lambda_i(D(H))$ this
set of strategies. With this definition, any strategy $\lambda_i \in
\Lambda_i(D(H))$ is defined on its outcomes in $H$, but when opposed to
any strategy $\lambda_{\m i} \in \Lambda_{\m i}(D(H))$, its outcomes
are in $D(H)$. Thus, when we iterate this operator, we do not need to
remember $H'$ and we can consider only the tree $D(H)$.
The tree $D(H)$ represents the strategy of $H$ that minimize the
regret in the following sense:

\begin{lemma}\label{prop:correctness:tree}
  Let $i\in\{1,2\}$. Let $\lambda_i \in \Lambda_i(H)$;
  $\minmax^{H}_i(\lambda_i) = \minmax^{H}_i$ iff $\lambda_i
  \in \Lambda_i(D(H))$.
\end{lemma}

Since there is a one-to-one correspondence between 
the strategies minimizing the regret in $G$ and
the strategies minimizing the minmax value in $H$, we
can define $D(G)$ by applying to  $D(H)$ the isomorphism
$\Phi^{-1}$, in other words by projecting the best alternatives away,
and by restoring the functions $\mu_i$. The set of strategies
$\Lambda_i(D(G))$ of $D(G)$ is defined as
$\Phi^{-1}(\Lambda_i(D(H)))$ (in other words, these are the strategies
of $D(H)$ where we project the best alternatives away). 
Let $\lambda_i\in\Lambda_i(G)$, by
Lemma~\ref{lem:mapstwo:tree}, it minimizes $\reg_i^G$ iff
$\Phi(\lambda_i)$ minimizes $\minmax_i^H$, and by
Lemma~\ref{prop:correctness:tree}, iff
$\Phi(\lambda_i)\in\Lambda_i(D(H))$, and finally, iff
$\lambda_i\in \Lambda_i(D(G))$. $D(G)$ represents the strategies of
$G$ minimizing the regret in the following sense:

\vspace{-2mm}
\begin{lemma}\label{lem:onestep:tree}
  Let $i\in\{1,2\}$. Let $\lambda_i \in \Lambda_i(G)$;
  $\reg^{G}_i(\lambda_i) = \reg^{G}_i$ iff $\lambda_i
  \in \Lambda_i(D(G))$.
\end{lemma}

We obtain a new tree $D(G)$ whose Player $i$'s strategies minimize the 
regret of Player $i$, for all $i=1,2$. We can iterate the regret computation on $D(G)$
and get the Player $i$'s strategies that minimize the regret of rank
$2$ of Player $i$, for all $i=1,2$. We continue iteration  we get a tree $G'$ such that $D(G') = G'$. 
We let $D^0(G)=G$ and $D^{j+1}(G)=D(D^j(G))$. Remind that 
$P_i^j$ are Player $i$'s strategies of $G$ that minimize the $j$-th regret.

\begin{proposition}\label{lem:iteratedregret:tree}
  Let $i\in\{1,2\}$ and $j>0$. We have $\reg^{G,j}_i =
  \reg^{D^{j-1}(G)}_i$ and $P^j_i = \Lambda_i(D^j(G))$.
\end{proposition}

\begin{proof}(sketch) By induction on $j$. It is clear 
  for $j=1$ and Lemma~\ref{lem:onestep:tree} ensures the
  correctness of the induction.
\end{proof}

\begin{theorem}
  Let $G = (S=S_1\uplus S_2,s_0,T,\mu_1,\mu_2,\target)$ be a tree
  arena. For all $i=1,2$, the iterated regret of Player $i$,
  $\reg^{G,\star}_i$, can be computed in $O(|S|^2)$.
\end{theorem}

\begin{proof}
  By Propositions~\ref{lem:stable} and~\ref{lem:iteratedregret:tree},
  there is an integer $j$ such that $\reg^{G,\star}_i =
  \reg^{D^j(G)}_i$. According to the definition of $D(G)$, $j\leq |S|$
  because we remove at least one edge of the tree at each step. Since
  $|D(G)|$ can be constructed in $O(|S|)$, the whole time complexity
  is $O(|S|^2)$.
\end{proof}

\subsection{IRM in Positive Weighted Arenas}
\label{sec:iterate-positive}

A weighted arena $G$ is said to be \textit{positive} if all
edges are weighted by strictly positive weights only. In this
section, we let $G = (S=S_1\uplus
S_2,s_0,T,\mu_1,\mu_2,\target_1,\target_2)$ be a positive weighted
arena. Remind that $P_i^j(G)$ is the set of strategies that minimize
$\reg_i^{G,j}$, for all $j\geq 0$ and $i=1,2$.

\begin{definition}[$j$-winning and $j$-bounded strategies]
  $\qquad$ Let $i\in\{1,2\}$ and $\lambda_i\in\Lambda_i(G)$. The strategy $\lambda_i$ is \textit{$j$-winning} if
  for all $\lambda_{\m i}\in P_{\m i}^j(G)$,
  $\outcome^G(\lambda_i,\lambda_{\m i})$ is winning. It is \textit{$j$-bounded} by some $B\geq 0$ if
  it is $j$-winning, and for all
  $\lambda_{\m i}\in P_{\m i}^j(G)$ and all $\kappa\in\{i,\m i\}$,  $\mu_\kappa(\outcome^{G,\target_i}(\lambda_i,\lambda_{\m i}))\leq
  B$.
\end{definition}

Note that $j$-boundedness differs from boundedness as we require
that the utilities of both players are bounded. We let $b^G = 6 (M^G)^3 |S|$.
We get a similar result than the boundedness of strategies
that miminize the regret of rank $1$, but for any rank:

\begin{lemma}\label{lem:boundedstrat}
    For all $i=1,2$ and  all $j\geq 0$, all $j$-winning strategies of Player $i$ 
    which minimize the $(j+1)$-th regret are $j$-bounded by $b^G$.
\end{lemma}

\begin{proof}(Sketch)
    First, if the regrets of first rank are infinite for both players, then 
    by definition of the iterated regret, $P_1^1=\Lambda_1(G)$ and $P_2^1=\Lambda_2(G)$ and
    thus their regrets are infinite at any rank. Therefore there is no winning strategy
    at any rank (otherwise one of the regrets would be finite).

    Suppose that the first regret of Player $i$ is finite for some $i=1,2$. By Lemma \ref{lem:regretbounded}, the winning strategies minimizing her first regret are bounded
    by $2M^G |S|$. Since the weights are strictly positive, the 
    lengths of the outcomes until $\target_i$ are bounded by $2M^G|S|$, which allows
    us to bound the utilities of Player $\m i$ until a first visit to $\target_i$ by $2(M^G)^2|S|$. Since $P_{i}^j(G)\subseteq P_i^1(G)$ for all 
    $j\geq 1$, the strategies of Player $i$ (which are necessarily winning as the regret is finite)
    at any rank are bounded by $2(M^G)^2|S|$. This bound is then used (non-trivially) to bound the winning strategies
    of Player $\m i$ by $6(M^G)^3|S|$. The full proof is in Appendix.
\end{proof}

Lemma \ref{lem:boundedstrat} allows us to reduce the problem to the iterated regret minimization
in a weighted tree arena, by unfolding the graph arena $G$ up to some
maximal utility value. Lemma \ref{lem:boundedstrat} suggests to take
$b^G$ for this maximal value. However the best responses to a strategy $j$-bounded by $b^G$ are not
necessarily bounded by $b^G$, but they are necessarily $j$-bounded by
$b^G\cdot M^G$, since the weights are strictly positive. Therefore we
let $B^G = b^G\cdot M^G$ and take $B^G$ as the
maximal value. Since the $j$-winning strategies are $j$-bounded by $b^G$ and the best responses
are $j$-bounded by $B^G$, we do not los
the set of finite plays $\pi$ of $G$ such that $\mu_G(\pi) \leq K$, for
all $i=1,2$. Note that $\play_K(G)$ is finite since $G$ has only
strictly positive weights. The unfolding of $G$ up to $B^G$ is naturally defined by
a tree weighted arena whose set of positions is $\play_{B^G}(G)$. 

\begin{definition}\label{def:reductutility}
    Let $G = (S=S_1\uplus
    S_2,s_0,T,\mu_1,\mu_2,\target_1,\target_2)$ be a positive weighted
    arena. The $B^G$-unfolding of $G$ is the weighted tree arena
    $G' = (S'=S'_1\uplus
    S'_2,s'_0,T',\mu'_1,\mu'_2,\target'_1,\target'_2)$
    such that $S'_i = \{ \pi\in\play_{B^G}(G)\ |\ \last(\pi)\in S_i\}$
    and for all $\pi,\pi'\in S'$, 
    $(\pi,\pi')\in T'$ iff $(\last(\pi),\last(\pi'))\in
    T$ and $\pi' = \pi.\last(\pi')$, and for all $i=1,2$, $\pi\in \target'_i$ iff
    $\last(\pi)\in\target_i$ and $\mu'_i(\pi,\pi') = 
    \mu_i(\last(\pi),\last(\pi'))$.
\end{definition}

We now prove that $\reg_{i}^{G,\star} = \reg_{i}^{G',\star}$, for all
$i=1,2$. As for edge-weighted arenas, this is done by defining
a surjective mapping $\Phi$ from $\Lambda_i(G)$ to
$\Lambda_i(G')$. For all $i=1,2$ and all $\lambda_i\in\Lambda_i(G)$,
and all $\pi\in \play_f(G)$ such that $\last(\pi)\in S_i$, 
$\Phi(\lambda_i)(\pi) = \perp$ if there is $\kappa\in\{1,2\}$ such
that $\mu_{\kappa}(\pi.\lambda_i(\pi)) > B^G$, and 
$\Phi(\lambda_i)(\pi) = \pi.\lambda_i(\pi)$ otherwise. 
This mapping is surjective, but not injective, since
two strategies that behave similarly up to some utility $B^G$ are
mapped to the same strategy.

\begin{lemma} For all $j\geq 1$, $\Phi(P_i^j(G)) = P_i^j(G')$ and 
  for all $\lambda_i\in P_i^j(G)$, $\reg_i^{G,j}(\lambda_i) =
      \reg_{i}^{G',j}(\Phi(\lambda_i))$.
\end{lemma}

This allows us to prove the correctness of the
reduction:

\begin{lemma}\label{lem:corres}
    For all $i = 1,2$, $\reg_{i}^{G,\star} = \reg_{i}^{G',\star}$.
\end{lemma}

\begin{proof}
  We prove that for all $j\geq 1$, $\reg_i^{G,j} = \reg_{i}^{G',j}$. Let $\lambda_i\in P_i^j(G)$. By definition of $P_i^j(G)$, $\lambda_i$ minimizes the
  $j$-th regret, so that $\reg_i^{G,j}(\lambda_i) = \reg_i^{G,j}$. By $(1)$, 
  $\reg_i^{G,j}(\lambda_i) = \reg_i^{G,j}(\Phi(\lambda_i))$. By $(2)$, 
  $\Phi(\lambda_i)\in P_i^j(G')$, therefore $\Phi(\lambda_i)$ minimizes the $j$-th regret in $G'$, so that
  $\reg_{i}^{G',j}(\Phi(\lambda_i)) = \reg_{i}^{G',j}$, from which we
  get $\reg_i^{G,j} = \reg_{i}^{G',j}$.
\end{proof}

By applying the algorithm of Section \ref{sec:iterate:trees} we get:

\begin{theorem}
    The iterated regret for both players
    in a positive weighted arena $G$ can
    be computed in pseudo-exponential time (exponential in $|S|$, $|T|$ and
    $M^G$).
\end{theorem}

For all $i=1,2$, the procedure of Section \ref{sec:iterate:trees} returns a
finite-memory strategy $\lambda_i$ minimizing the iterated regret in $G'$
whose memory is the best alternatives seen so far by both players.
From $\lambda_i$ we can compute a finite-memory strategy in $G$
minimizing the iterated regret of Player $i$, the needed memory is
the best alternatives seen by both players and the current finite play
up to $B^G$. When the utility is greater than $B^G$, then any move 
is allowed. Therefore one needs to add one more bit of memory
expressing whether the utility is greater than $B^G$.

Finally, the unfolding of the graph arena up to $B^G$ is
used to finitely represent the (potentially infinite) sets of
strategies of rank $j$ in $G$. Finding such a representation 
is not obvious for the full class of weighted arenas, since before
reaching its objective, a player can take a $0$-cost loop finitely
many times without affecting her minimal regret. This suggests 
to add \textit{fairness} conditions on edges to compute the iterated
regret. This is illustrated by the following example.

\begin{figure}[h]
\unitlength=0.8mm
\vspace{-72mm}
\centering{



\gasset{linewidth=0.14,Nw=8.0,Nh=8.0,Nmr=4.0,Nadjustdist=1.0,ilength=5.0,flength=5.0,rdist=0.7,loopdiam=8.0,AHdist=1.41,AHLength=1.5,AHlength=1.41,ELdist=1.0}
\begin{picture}(105,153)(0,-153)
\node[NLangle=0.0,Nmarks=i](n0)(24.0,-104.0){A}

\node[NLangle=0.0,Nmr=0.0](n1)(52.0,-104.0){B}

\node[NLangle=0.0,Nmarks=r](n2)(36.0,-128.0){C}

\drawedge[ELdist=0.5,curvedepth=8.0](n0,n1){$0/0$}

\drawedge[ELdist=0.5,ELside=r,curvedepth=8.0](n1,n0){$0/0$}

\drawedge[ELside=r,ELdist=1.0](n0,n2){$5/0$}

\drawedge[ELdist=1.1](n1,n2){$0/5$}

\node[NLangle=0.0,Nmarks=i,Nmr=0.0](n7)(72.0,-104.0){D}

\node[NLangle=0.0](n8)(96.0,-104.0){E}

\node[NLangle=0.0,Nmarks=r](n9)(72.0,-128.0){F}

\drawedge[ELdist=0.5,curvedepth=8.0](n7,n8){$0/0$}

\drawedge[ELdist=0.5,ELside=r,curvedepth=8.0](n8,n7){$0/0$}

\drawedge[ELside=r,ELdist=1.0](n7,n9){$0/0$}

\end{picture}
}
\vspace{-20mm}
  \caption{\label{fig:ex3} Free loops}
\end{figure}
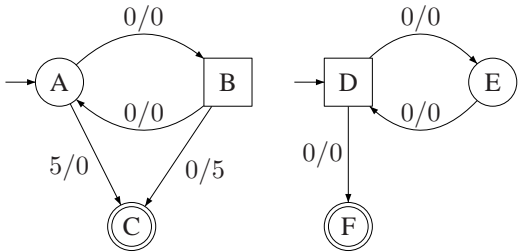


\begin{example}
Consider the left example of Fig.~\ref{fig:ex3}.
Player $1$'s strategies minimizing the regret are 
those that pass finitely many times by the edge $(A,B)$ and finally
move to $C$. The regret is therefore $5$. Similarly, the strategies
minimizing Player $2$'s regret are those that pass finitely many
times by $(B,A)$ and finally move to $C$. The regret is $5$ as well.
The regret of rank $2$ for Player $1$ is $5$ as well, and
the set of strategies minimizing it is also the same as before (and
similary for Player $2$). Indeed, the regret of a Player $1$'s
strategy that passes $K$ times by $(A,B)$ is $5$, since Player $2$ can
maximize her regret with a strategy that passes at least $K$ times by
$(B,A)$. Thus $\reg_1^{G,\star} = \reg_2^{G,\star} = 5$.

On the right example, Player $1$ has no winning strategy at the first
rank and her regret is $+\infty$. However the strategies of Player $2$
minimizing her regret are the ones that pass finitely many times
through the loop. Therefore all the strategies of Player $1$ are winning at rank $2$. The iterated
regret of both players is $0$.
\end{example}



\section{Conclusion} The theory of infinite qualitive non-zero sum
games over graphs is still in an initial development stage. We
adapted a new solution concept from strategic games to game graphs,
and gave algorithms to compute the regret and iterated regret. The
strategies returned by those algorithms have a finite memory.  One
open question is to know whether this memory is necessary. In other words,
are memoryless strategies sufficient to minimize the (iterated) regret
in game graphs? Another question is to determine the lower bound
on the complexity of (iterated) regret minimization. Iterated regret
minimization over the full class of graphs is still open. Finally, we
think that this work can easily be extended to an $n$-player setting. 


\bibliographystyle{plain}

\begin{thebibliography}{1}

\bibitem{Basu94}
Kaushik Basu.
\newblock The traveler's dilemma: Paradoxes of rationality in game theory.
\newblock {\em American Economic Review}, 84(2):391--95, 1994.

\bibitem{Brihaye10}
Thomas Brihaye, V\'eronique Bruy\`ere, and Julie~De Pril.
\newblock Equilibria in quantitative reachability games.
\newblock 2010.
\newblock submitted, available at \texttt{http://www.ulb.ac.be/di/ssd/cfv/}.

\bibitem{Chatterjee04gameswith}
Krishnendu Chatterjee, Thomas~A. Henzinger, and Marcin Jurdzinski.
\newblock Games with secure equilibria.
\newblock In {\em LICS}, pages 160--169, 2004.

\bibitem{Kupf10}
Dana Fisman, Orna Kupferman, and Yoad Lustig.
\newblock Rational synthesis.
\newblock 2010.
\newblock to appear in TACAS'10.

\bibitem{GraedelTW02}
Erich Gr{\"a}del, Wolfgang Thomas, and {Th}omas Wilke, editors.
\newblock {\em Automata, Logics and Infinite Games}, volume 2500 of {\em
  Lecture Notes in Computer Science}.
\newblock Springer, 2002.

\bibitem{Ummels08}
Erich Gr\"adel and Michael Ummels.
\newblock Solution concepts and algorithms for infinite multiplayer games.
\newblock In {\em New Perspectives on Games and Interaction}, volume~4 of {\em
  Texts In Logic and Games}, pages 151--178, 2008.

\bibitem{HalpernPass09}
Joseph~Y. Halpern and Rafael Pass.
\newblock Iterated regret minimization: A more realistic solution concept.
\newblock In {\em IJCAI}, 2009.

\bibitem{Osborne94}
M.J. Osborne and A.~Rubinstein.
\newblock {\em A Course in Game Theory}.
\newblock MIT Press, 1994.

\bibitem{Rosenthal82}
Robert~W. Rosenthal.
\newblock Games of perfect information, predatory pricing and the chain-store
  paradox.
\newblock {\em Journal of Economic Theory}, 25(1):92--100, 1981.

\end{thebibliography}

\newpage
\section{Appendix}

\subsection{Missing Proofs of Section \ref{sec:targetgraph}}

\paragraph{Proposition \ref{prop:ba}}

\begin{proof}
Proof by induction on $|\pi|$. 

If $|\pi|= 0$, then $\ba_1^{G'}(\pi)= +\infty = b_0$.  

We now assume that the property is true for any finite path $\pi$ in
$G'$ from $s'_0$ to some $(s,b)$ of length $k$. Let
$\pi=(s_0,b_0)\dots(s_k,b_k)(s_{k+1},b_{k+1})$ be a path of length
$k+1$. We have:

$$
\begin{array}{rcl}
& & \ba_1^{G'}(\pi) \\
  &=&  \min_{0\leq j<k+1} \ba_1^{G'}(s'_j,s'_{j+1})\\
&=& \min  ( \min_{0\leq j<k} \ba_1^{G'}(s'_j,s'_{j+1}), \ba_1^{G'}(s'_k,s'_{k+1}))\\
&=& \min  (b_k, \ba_1^{G'}(s'_k,s'_{k+1})) \mbox{ by induction hypothesis} \\
&=& \min  (b_k, \ba_1^{G}(s_k,s_{k+1})) \mbox{ by ($\star$)} \\
&=& b_{k+1}  \mbox{ by definition of $G'$}
\end{array}
$$
($\star$) According to Definition~\ref{def:reduc1}, $\forall (s,b) \in
\target'_1 :\ \mu'_1(s,b) = \mu_1(s)$. Thus $\forall (s,b)\in
S',\best^{G'}_1((s,b)) =\best^G_1(s)$ and $\forall
\left((s,b),(s',b')\right)\in T',
\ba_1^{G'}\left((s,b),(s',b')\right) = \ba_1^{G}(s,s')$.
\end{proof}

\paragraph{Proposition \ref{prop:construction_G'}}

\begin{proof}
  Constructing $G'$ is done in three steps: 
  \begin{enumerate}

  \item compute all the values $\best_1^G(s)$, for all $s\in S$; this
    step is equivalent to looking for the shortest path to the
    objective and has a complexity of $O(log_2(M_1^G)(|S|+|T|))$.

  \item compute all the values $\ba^G_1(s,s')$, for all $(s,s')\in T$
    such that $s\in S_1$; it can be computed with a time complexity $O(|T|)$

  \item construct $G'$ by a fixpoint algorithm; this graph has at most $|\target_1|\times |S|$ states and $|\target_1|\times |T|$ transitions.
  \end{enumerate}
\end{proof}

\paragraph{Lemma \ref{lem:maxmax}}

\begin{proof}
  Let $\lambda_1\in\Lambda_1(G')$. If $\lambda_1$ is losing, there is
  a strategy $\lambda_2\in\Lambda_2(G')$ such that
  $\outcome^{G'}(\lambda_1,\lambda_2)$ is losing. 
  Therefore $\reg_1^{G'}(\lambda_1) = \reg_1^{G'}(\lambda_1,\lambda_2)
  = +\infty = \nu_1(\lambda_1,\lambda_2) =
  \max_{\lambda_2\in\Lambda_2(G')} \nu_1(\lambda_1,\lambda_2)$.

  Suppose that $\lambda_1$ is a winning strategy and let $\lambda_2\in\Lambda_2(G')$
  which maximizes $\reg_1^{G'}(\lambda_1,\lambda_2)$. Let $\pi =
  s_0s_1\dots s_n = \outcome^{G',\target_1}(\lambda_1,\lambda_2)$. We define a
  strategy $\lambda_2'$ that plays as $\lambda_2$ on $\pi$ and
  cooperates with Player $1$ if she would have deviated from $\pi$. Formally, for all $h\in
  \play(G')$ such that $\last(h)\in S'_2$, we let 
  $\lambda'_2(h) = s_{j+1}$ if there is $j<n$ such that $h =
  s_0s_1\dots s_j$. Otherwise we let 
  $\lambda'_2(h) = s$ such that $(\last(h),s)\in T'$ and
  $\best^{G'}_1(s)$ is minimal (among the successors of $\last(h)$).

  \noindent Clearly, $\pi = \outcome^{G',\target_1}(\lambda_1,\lambda'_2)$ and
  $\br^{G'}_1(\lambda_2')\leq \br^{G'}_1(\lambda_2)$.  Therefore
  $\reg_1^{G'}(\lambda_1,\lambda_2)\leq
  \reg_1^{G'}(\lambda_1,\lambda'_2)$. Since $\lambda_2$ maximizes the
  regret, we get $\reg_1^{G'}(\lambda_1,\lambda_2) =
  \reg_1^{G'}(\lambda_1,\lambda'_2)$.

  \noindent The best response to $\lambda'_2$ either deviates from $\pi$ or not. If the
  best response deviates from $\pi$ at a node $s_j$, $j<n$,
  i.e. chooses a node $s'$ such that $s'\neq s_{j+1}$, then the utility 
  of the best response, according to the definition of $\lambda'_2$,
  is $\best_{G'}^1(s')$. The best response to $\lambda'_2$ minimizes over all those
  possibilities, therefore $\br^{G'}_1(\lambda_2') = \min(\mu'_1(s_n),\min_{j<n,(s_j,s')\in T', s'\neq s_{j+1}}
  \best^{G'}_1(s'))$,
  i.e. $\min(\mu'_1(s_n),\ba^{G'}_1(\pi))$. By Proposition
  \ref{prop:ba}, $\ba^{G'}_1(\pi) = \proj_2(s_n)$. Therefore $\reg_1^{G'}(\lambda_1,\lambda_2) =
  \reg_1^{G'}(\lambda_1,\lambda'_2) = \mu'_1(s_n) -
  \min(\mu'_1(s_n),p_2(s_n)) = \nu_1(s_n)$. From which we get
  $\reg^{G'}_1(\lambda_1) \leq \max_{\lambda_2\in\Lambda_2(G')}
  \nu_1(\lambda_1,\lambda_2)$.

  Conversely, let $\lambda_2$ which maximizes
  $\nu^1(\lambda_1,\lambda_2)$. Since $\lambda_1$ is winning, 
  we can define $\pi =
  \outcome^{G',\target'_1}(\lambda_1,\lambda_2)$.
  Similarly as forth direction of the proof, one can construct
  a strategy $\lambda'_2$ that plays like $\lambda_2$ along $\pi$ and
  cooperates with Player $1$ when deviating from $\pi$. Clearly, this strategy
  has the same outcome as $\lambda_2$ and we get
  $\reg_1^{G'}(\lambda_1,\lambda'_2) =
  \nu_1(\lambda_1,\lambda_2)$. Finally we have
  $\reg_1^{G'}(\lambda_1) \geq \reg^{G'}_1(\lambda_1,\lambda'_2) =
  \nu_1(\lambda_1,\lambda_2) = \max_{\lambda_2}
  \nu_1(\lambda_1,\lambda_2)$.
\end{proof}

\paragraph{Lemma \ref{lem:bij}}

\begin{proof}
    The mapping $\Phi$ has been defined in the paper. 
    It remains to prove that $\reg_G^1(\lambda_1) = \reg^{G'}_1(\Phi_1(\lambda_1))$, for all
    $\lambda_1\in\Lambda_1(G)$.

    For all $\lambda_1\in\Lambda_1(G)$, all $\lambda_2\in\Lambda_2(G)$,
    $\outcome^G(\lambda_1,\lambda_2) = \proj_1(\outcome^{G'}(\Phi(\lambda_1),\Phi(\lambda_2)))$.
    Therefore $\ut_1^G(\lambda_1,\lambda_2) = \ut_1^{G'}(\Phi(\lambda_1),\Phi(\lambda_2))$.

    Finally:

    $$
    \begin{array}{lllllll}
        & &\!\!\!\! \reg_1^G(\lambda_1)\\
        & = &\!\!\!\! \max\limits_{\lambda_2\in\Lambda_2(G)} \ut_1^G(\lambda_1,\lambda_2) - \min\limits_{\lambda_1^*\in\Lambda_1(G)}\ut_1^G(\lambda_1^*,\lambda_2) \\
        & = &\!\!\!\! \max\limits_{\lambda_2\in\Lambda_2(G)} \ut_1^{G'}(\Phi(\lambda_1),\Phi(\lambda_2)) - \min\limits_{\lambda_1^*\in\Lambda_1(G)}\ut_1^{G'}(\Phi(\lambda_1^*),\Phi(\lambda_2)) \\
        & = &\!\!\!\! \max\limits_{\lambda_2\in\Lambda_2(G')} \ut_1^G(\Phi(\lambda_1),\lambda_2) - \min\limits_{\lambda_1^*\in\Lambda_1(G')}\ut_1^G(\lambda_1^*,\lambda_2) \\
        &  &\!\!\!\! \text{ (since $\Phi(\Lambda_i(G)) = \Lambda_i(G')$ for all $i=1,2$)} \\
        & = &\!\!\!\! \reg_1^{G'}(\Phi(\lambda_1))
    \end{array}
    $$

\end{proof}

\subsection{Missing Proofs of Section \ref{sec:graphs}}

\paragraph{Lemma \ref{lem:preserveregret}}

The proof of this lemma is supported by the following lemma, which says that
under certain conditions, the utility of the outcomes in $G$ and $G'$ are equal
modulo $\Phi$:

\begin{lemma}\label{lem:maputilities} Let $\lambda_1\in\Lambda_1(G)$ and
  $\lambda_2\in\Lambda_2(G)$. If $\outcome^{G'}(\Phi(\lambda_1),\Phi(\lambda_2))$ is winning for
  Player $1$ in $G'$ or $\ut_{1}^G(\lambda_1,\lambda_2)\leq B$, then $\ut^G_1(\lambda_1,\lambda_2) =
  \ut_1^{G'}(\Phi(\lambda_1),\Phi(\lambda_2))$.
\end{lemma}

\begin{proof}
  If $\ut_1^G(\lambda_1,\lambda_2)\leq B$, then
  $\outcome^G(\lambda_1,\lambda_2)$ is winning, and we let 
  $\pi = \outcome^{G,\target_1}(\lambda_1,\lambda_2)$. We enrich $\pi$
  with the utilities of Player $1$ by defining a path $\pi' =
  (s_0,u_0)\dots (s_n,u_n)$ where $\pi = s_0\dots s_n$ and for all
  $j\leq n$, $u_j = \mu_1^G(s_0\dots s_j)$. Since $\pi$ is
  bounded, we have $u_j\leq B$ for all $j\leq n$, and by definition of
  $G'$, $\pi'$ is a path of $G'$.   By definition of $\Phi$ we
  clearly have $\pi' = \outcome^{G'}(\Phi(\lambda_1),\Phi(\lambda_2))$, from
  which we get $\ut_1^G(\lambda_1,\lambda_2) = 
  \ut_1^{G'}(\Phi(\lambda_1),\Phi(\lambda_n))$.

  If $\outcome^{G'}(\Phi(\lambda_1),\Phi(\lambda_2))$ is winning for
  Player $1$, we let $\pi' =   \outcome^{G',\target'_1}(\Phi(\lambda_1),\Phi(\lambda_2))$ and $\pi = \proj_1(\pi')$. Clearly, $\pi'$ is
  a winning play for Player $1$ in $G$ and by definition of $\Phi$,
  $\pi' = \outcome^{G,\target_1}(\lambda_1,\lambda_2)$, from which we
  get the equality of the utilities.
\end{proof}

We can now prove Lemma \ref{lem:preserveregret}:

\begin{proof}(Proof of Lemma \ref{lem:preserveregret})
  Let $\lambda_2\in \Lambda_2(G)$ which maximizes 
  $\reg_1^G(\lambda_1,\lambda_2)$, and $\lambda_1^*$ be the best
  response to $\lambda_2$. Therefore $\reg_1^G(\lambda_1) =
  \ut_1^G(\lambda_1,\lambda_2) - \ut_1^G(\lambda_1^*,\lambda_2)$. 
  Since $\lambda_1$ is bounded by $B$, we have
  $\ut_1^G(\lambda_1,\lambda_2)\leq B$ and
  $\ut_1^G(\lambda_1^*,\lambda_2)\leq B$ (since $\lambda_1^*$ is
  at least as good as $\lambda_1$). By Lemma \ref{lem:maputilities}, we get
  $\ut_1^G(\lambda_1,\lambda_2) =
  \ut_1^{G'}(\Phi(\lambda_1),\Phi(\lambda_2))$ and
  $\ut_1^G(\lambda^*_1,\lambda_2) =
  \ut_1^{G'}(\Phi(\lambda^*_1),\Phi(\lambda_2))$.

  By definition of the best response, $\br^{G'}_1(\Phi(\lambda_2))\leq
  \ut_1^{G'}(\Phi(\lambda^*_1),\Phi(\lambda_2))$.
  Therefore 

  $$
  \begin{array}{llllllll}
    \reg_1^G(\lambda_1) & = & \reg_1^G(\lambda_1,\lambda_2) \\
    & = & \ut_1^G(\lambda_1,\lambda_2) -
    \ut_1^G(\lambda_1^*,\lambda_2) \\
& = & \ut_1^{G'}(\Phi(\lambda_1),\Phi(\lambda_2)) -
    \ut_1^{G'}(\Phi(\lambda_1^*),\Phi(\lambda_2)) \\
    & \leq & \ut_1^{G'}(\Phi(\lambda_1),\Phi(\lambda_2)) -
    \br^{G'}_1(\Phi(\lambda_2)) \\
    & = & \reg_1^{G'}(\Phi(\lambda_1),\Phi(\lambda_2)) \\
    & \leq & \reg_1^{G'}(\Phi(\lambda_1))
  \end{array}
  $$

  Conversely, since $\Phi(\Lambda_2(G)) = \Lambda_2(G')$, there exists
  $\lambda_2\in\Lambda_2(G)$ such that $\Phi(\lambda_2)$ maximizes
  $\reg_1^{G'}(\Phi(\lambda_1),\Phi(\lambda_2))$. Similarly, there is
  $\lambda_1^*\in\Lambda_1(G)$ such that $\Phi(\lambda_1^*)$ is the
  best response to $\Phi(\lambda_2)$. Since $\lambda_1$ is bounded by
  $B$, we have $\ut_1^G(\lambda_1,\lambda_2)\leq B$, and by Lemma
  \ref{lem:maputilities}  we get $\ut_1^G(\lambda_1,\lambda_2) =
  \ut_1^{G'}(\Phi(\lambda_1),\Phi(\lambda_2))$. Since $\lambda_1$ is
  winning for Player $1$ and bounded by $B$ in $G$, $\Phi(\lambda_1)$ is also winning for
  Player $1$ in $G'$. Therefore $\Phi(\lambda_1^*)$ is also winning
  for Player $1$ in $G'$ (since it does at least as good as
  $\Phi(\lambda_1)$ against $\Phi(\lambda_2)$). Therefore
  $\outcome^{G'}(\Phi(\lambda_1),\Phi(\lambda_2))$ is winning for
  Player $1$ in $G'$, and by Lemma \ref{lem:maputilities} we get
  $\ut_1^G(\lambda_1^*,\lambda_2) =
  \ut_1^{G'}(\Phi(\lambda_1^*),\Phi(\lambda_2))$. Finally:

  $$
  \begin{array}{llllllll}
    & & \reg_1^{G'}(\Phi(\lambda_1)) \\
    & = & \reg_1^{G'}(\Phi(\lambda_1),\Phi(\lambda_2)) \\
    & = & \ut_1^{G'}(\Phi(\lambda_1),\Phi(\lambda_2)) -
    \ut_1^{G'}(\Phi(\lambda_1^*),\Phi(\lambda_2)) \\
    & = & \ut_{1}^G(\lambda_1,\lambda_2) -
    \ut_{1}^G(\lambda_1^*,\lambda_2) \\
    & \leq & \ut_{1}^G(\lambda_1,\lambda_2) -
    \br^{G}_1(\lambda_2) \\
    & = & \reg_{1}^G(\lambda_1,\lambda_2) \\
    & \leq & \reg_{1}^G(\lambda_1)
  \end{array}
  $$
\end{proof}

\subsection{Missing Proofs of Section \ref{sec:iterate}}




\paragraph{Lemma \ref{lem:mapstwo:tree}}

By projecting away the best alternatives of Player $\m i$ in $H$, we
get a tree isomorphic to $H$ which corresponds exactly to the
tree of best alternatives defined in Section~\ref{sec:targetgraph},
in which all the results stated in Lemma \ref{lem:mapstwo:tree} have
been already proved. Clearly, adding the best alternatives of the other player
does not change those results.

\paragraph{Lemma \ref{prop:correctness:tree}}

\begin{proof}
  Let $\lambda_i\in\Lambda_i(H)$ such that $\minmax^{H}_i(\lambda_i)
  = \minmax^{H}_i$. Let $s\in S'_i$ be a position of $D(H)$
  compatible with $\lambda_i$, i.e. such that there is
  $\lambda_{\m i}\in\Lambda_{\m i}(D(H))$ such that $s$
  occurs in $\outcome^{D(H)}(\lambda_i,\lambda_{\m i})$. 
  Let $s'=\lambda_i(s)$. We have to prove
  that $s'$ is a position of $D(H)$. We have
  $\minmax^{H}_i(s') \leq \minmax^{H}_i(\lambda_i)=
  \minmax^{H}_i$. Indeed, since Player $\m i$ is able
  to enforce Player $i$ to go to $s'$ when she plays $\lambda_i$, if $\minmax^{H}_i(s') > \minmax^{H}_i$, 
  then $\lambda_i$ does not minimize $\minmax^H_i$. 
  According to the definition of the delete operator,
  $(s,s')$ is an edge of $D(H)$. Thus $s'$ is a position of $D(H)$,
  and $\lambda_i\in\Lambda_i(D(H))$.

  Conversely, if $\lambda_i \in \Lambda_i(D(H))$. We proceed
  {\it reductio ad absurdum}.

  If $\minmax^{H}_i(\lambda_i)> \minmax^{H}_i$, there exists $\lambda_{\m
    i}\in\Lambda_{\m i}(H)$ such that
  $\minmax^{H}_i(\lambda_i,\lambda_{\m i})= \minmax^{H}_i(\lambda_i)>
  \minmax^{H}_i$. We let $\pi=\pi_0\dots
  \pi_n = \outcome^{H}(\lambda_i,\lambda_{\m i})$. Let $s,b_1,b_2$ such that 
  $\pi_n = (s,b_1,b_2)$.  Since $\pi_n\in\target'$, 
  $\minmax^{H}_i(\lambda_i)=\mu'_i(\pi_n)=\minmax^{H}_i(\pi_n)$.
  We consider the first position $\pi_k$ along $\pi$ such that $k<n$ and $\pi_k$
  is owned by Player $i$, i.e. $\pi_k\in S'_i$ and $\pi_{k+1}\dots
  \pi_{n-1} \in (S'_{\m i})^*$. This position exists, otherwise 
  all positions $\pi_0,\dots, \pi_{n-1}$ are owned by Player $\m i$, 
  and therefore $\minmax_i^H \geq \mu_i'(\pi_n) = \minmax_i^H(\lambda_i)$, which contradicts
  our hypothesis. Since $\lambda_i\in\Lambda_i(D(H))$, by definition
  of $\Lambda_i(D(H))$,  $(\pi_k,\pi_{k+1})$ is an edge of $H'$. 
  Since from $\pi_{k+1}$ Player $\m i$ can enforce Player $i$
  to go to $\pi_n$ (since there are only positions owned by Player
  $\m i$ along $\pi_k\dots \pi_{n-1}$), we have
  $\minmax^{H}_i(\pi_{k+1})\geq \minmax^{H}_i(\pi_{n})>
  \minmax^{H}_i$. Since $\pi_k\in S'_i$, this contradicts the definition of $H'$ (and $D(H)$), 
  because the edge $(\pi_k,\pi_{k+1})$ would have been removed.
  Thus $\minmax^{H}_i(\lambda_i) = \minmax^{H}_i$.
\end{proof}


\paragraph{Proposition \ref{lem:iteratedregret:tree}}

\begin{proof}Proof by induction on $j$.

  If $j=1$, we have $\reg^{G,1}_i = \reg^{D^{0}(G)}_i = \reg^G_i$ and
  by Lemma~\ref{prop:correctness:tree}, $P^1_i=\Lambda_i(D(G))$.

  We assume that $\reg^{G,j}_i = \reg^{D^{j-1}(G)}_i$ and $P^j_i=
  \Lambda_i(D^{j}(G))$. By definition, $\reg^{G,j+1}_i =
  \reg^{G,P^j_i,P^j_{\m i}}_i$. By induction hypothesis, $P^j_i=
  \Lambda_i(D^{j}(G))$, thus $\reg^{G,P^j_i,P^j_{\m i}}_i =
  \reg^{D^j(G)}_i = \reg^{G,j+1}_i$.

  Moreover, $\lambda_i\in P^{j+1}_i$ iff $\lambda_i\in P^{j}_i$ and
  $\reg^{G,j+1}_i(\lambda_i)= \reg^{G,j+1}_i$. By induction
  hypothesis, $P^j_i= \Lambda_i(D^{j}(G))$. We demonstrated that
  $\reg^{G,j+1}_i = \reg^{D^j(G)}_i$. By
  Lemma~\ref{prop:correctness:tree}, applied to the tree $D^{j}(G)$,
  we have $\reg^{D^j(G)}_i(\lambda_i)= \reg^{D^j(G)}_i$ iff $\lambda_i
  \in \Lambda_i(D(D^{j}(G)))$. By definition,
  $\reg^{D^j(G)}_i(\lambda_i)= \max_{\lambda_{\m i}\in \Lambda_{\m
      i}(D^j(G))} \reg^{D^j(G)}_i(\lambda_i,\lambda_{\m i}) =
  \max_{\lambda_{\m i}\in \Lambda_{\m i}(D^j(G))} \left(
    \ut^{G}_i(\lambda_i,\lambda_{\m i})- min_{\lambda_{i}^*\in
      \Lambda_{ i}(D^j(G))} \ut^{G}_i(\lambda_i^*,\lambda_{\m
      i})\right) $. Since $P^j_i= \Lambda_i(D^{j}(G))$, we have
  $\reg^{D^j(G)}_i(\lambda_i)= \reg^{G,j+1}_i(\lambda_i)$. So
  $\reg^{G,j+1}_i(\lambda_i)= \reg^{G,j+1}_i$ iff
  $\reg^{D^j(G)}_i(\lambda_i)= \reg^{D^j(G)}_i$.

  Consequently, $\lambda_i\in P^{j+1}_i$ iff $\lambda_i \in
  \Lambda_i(D^{j+1}(G))$.
\end{proof}

\subsection{Missing Proofs of Section \ref{sec:iterate-positive}}

In this section, we prove several lemmas that do not appear in
the paper, especially to prove Lemma \ref{lem:corres}.

\paragraph{Lemma \ref{lem:boundedstrat}}

\begin{proof}
    Suppose that there is no winning strategy for both players in $G$. Therefore 
    $\reg_1^G = \reg_{1}^{G,1} = \reg_2^G = \reg_{2}^{G,1} = +\infty$, $P_1^1 = \Lambda_1(G)$ and $P_2^1 = \Lambda_2(G)$. 
    It is easy to verify that there is no $j$-winning strategy for
    both players and all ranks $j$.

    Suppose that Player $i$ has a winning strategy, for some $i=1,2$.
    Therefore by Lemma \ref{lem:regretbounded}, the strategies
    minimizing the regret are bounded by $2 M^G |S|$. 
    Since $S^i_{j}\subseteq S^i_{0} = \Lambda_i(G)$ for all $j\geq 0$, we get
    that all strategy of $S^i_j$ is bounded by $2 M^G
    |S|$. Let $j\geq 0$, $\lambda_i\in S^i_j$ and $\lambda_{\m i}\in
    \Lambda_{\m i}(G)$. Let $\pi =
    \outcome^{G,\target_i}(\lambda_i,\lambda_{\m i})$. We have
    $\mu_i(\pi)\leq 2 M^G |S|$, and
    since the weights are strictly positive integers,
    $|\pi|\leq 2 M^G |S|$. Therefore $\mu_{\m i}(\pi)\leq
    2(M^G)^2|S|$. In other words, for all $j\geq 0$, all
    strategy of $P_i^j$ is $j$-winning and $j$-bounded by
    $2(M^G)^2|S|$.

    It remains to prove that the
    $j$-winning strategies of Player $\m i$ minimizing the $(j+1)$-th
    regret are also $j$-bounded.
    Let $j_0\geq 0$ be the first natural number such that $P^{j_0}_{\m
      i}$ contains a $j_0$-winning strategy (if it exists).
    If $j_0 = 0$, then $P^{j_0}_{\m i} = \Lambda_{\m i}(G)$. If
    $j_0>0$, then no strategy of $P^{j_0-1}_{\m i}$ is
    $(j_0-1)$-winning by definition of $j_0$, so
    that $\reg_{\m i}^{G,j_0} = +\infty$, from which we get
    $P^{j_0}_{\m i} = \Lambda_{\m i}(G)$. In both cases, we have
    $P^{j_0}_{\m i} = \Lambda_{\m i}(G)$.

    Since after reaching her objective, Player $i$ can play
    however she wants without affecting her regret, 
    there is a strategy $\gamma_{\m i}\in \Lambda_{\m i}(G)$  that
    wins against all strategies of $P^{j_0}_i$ and which is
    memoryless once Player $i$ has reached his objective.
    Formally, there is a memoryless strategy $\gamma'_{\m i} : S_{\m
      i}\rightarrow S$ such
    that for all $\pi\in\play_f(G)$ such that $\last(\pi)\in S_{\m i}$, 
    if $\pi$ contains a position of $\target_{\m i}$, then
    $\gamma_{\m i}(\pi) = \gamma'_{\m i}(\last(\pi))$.

    Let $\lambda_i\in P^{j_0}_i$. We
    now bound the size of $\outcome^{G,\target_{\m i}}(\gamma_{\m i},\lambda_i)$, which will provide a
    bound on the utility. Let $\pi_{\m i} =
    \outcome^{G,\target_{\m i}}(\gamma_{\m i},\lambda_i)$ and $\pi_i =
    \outcome^{G,\target_i}(\lambda_i,\gamma_{\m i})$. We consider two
    cases:

    \begin{itemize}
      \item if $\pi_{\m i}$ is a prefix of $\pi_i$. We already know that
        $\lambda_i$ is $j$-bounded by $2 (M^G)^2 |S|$,
        therefore we also get $\mu_\kappa(\pi_2)\leq \mu_\kappa(\pi_1) \leq 
        2 (M^G)^2 |S|$, for all $\kappa=1,2$;

      \item if $\pi_i$ is a prefix of $\pi_{\m i}$, then $\pi_{\m i} =
        \pi_i\pi_i'$, for some $\pi_i'$. Since $\lambda_i$ is
        $j$-bounded by $2 (M^G)^2 |S|$,
        $\mu_\kappa(\pi_i)\leq 2 (M^G)^2 |S|$, for all
        $\kappa=1,2$. Since $\gamma_{\m i}$ is
        memoryless after $\pi_i$, there is no loop in
        $\pi_i'$. Therefore $\mu_\kappa(\pi_i') \leq |S| M^G$, for
        all $\kappa=1,2$. Finally, for all $\kappa=1,2$, we get
        $\mu_\kappa(\pi_{\m i}) = \mu_\kappa(\pi_i) + \mu_\kappa(\pi'_i) \leq 
        (2 (M^G)^2 + M^G) |S| \leq 3 (M^G)^2 |S|$.

    \end{itemize}

    In both cases, we get $\mu_\kappa(\gamma_{\m i},\lambda_i)\leq
    3 (M^G)^2 |S|$, for all $\kappa=1,2$ and all $\lambda_i\in P_i^j$.
    Therefore $0\leq \br^{G}_{\m i}(\lambda_i)\leq 3 (M^G)^2 |S|$ ($\star$), which holds
    for all $\lambda_i\in P^{j_0}_i$. We also get $\reg_{\m i}^{G,j_0+1} \leq \reg_{\m i}^{G,j_0+1}(\gamma_{\m i})\leq 3 (M^G)^2 |S|$ ($\star\star$).

    Let now $\lambda_{\m i}$ which minimizes $\reg_{\m i}^{G,j_0+1}$ and
    $\lambda_i\in P^{j_0}_i$. Let $\pi = \outcome^{G,\target_{\m
        i}}(\lambda_{\m i},\lambda_i)$.
    By $(\star\star)$, we have
    $\reg_{\m i}^{G,j_0+1}(\lambda_{\m i},\lambda_i)\leq 3 (M^G)^2 |S|$,
    ie

    $$
    \mu_{\m i}(\pi) - \br^{G,P^{j_0}_{\m i}}_{\m i}(\lambda_i)\leq 2(M^G)^2|S|
    $$

    \noindent Since $P^{j_0}_{\m i} = \Lambda_{\m i}(G)$,
    $\br^{G,P^{j_0}_{\m i}}_{\m
        i}(\lambda_i) = \br^{G}_{\m i}(\lambda_i)$. Therefore
    by $(\star)$, we get $\br^{G,P^{j_0}_{\m i}}_{\m i}(\lambda_i)\leq
    3 (M^G)^2 |S|$, and $\mu_{\m i}(\pi)\leq
    6 (M^G)^2 |S|$. The weights being strictly
    positive, we get $\mu_i(\pi)\leq 6 (M^G)^3 |S| = b^G$.

    Therefore all $j_0$-winning strategy of Player $\m i$
    which minimizes the $(j_0+1)$-th regret is $j_0$-bounded by
    $b^G$, and \textit{a fortiori} 
    all $j$-winning strategy is also $j_0$-bounded by
    $b^G$, for all $j\geq j_0$.
\end{proof}

\begin{lemma}\label{lem:mapbounded}
  Let $i=1,2$, $j\geq 0$, $\lambda_i\in P_i^j(G)$
  and $\lambda_{\m i}\in P_{\m i}^j(G)$. Let $o =
  \outcome^G(\lambda_i,\lambda_{\m i})$. If $o$ is winning for
  Player $i$, $\mu_i(o) \leq B^G$
  and $\mu_{\m i}(o) \leq B^G$, then
  $\ut_i^G(\lambda_i,\lambda_{\m i}) =
  \ut_i^{G'}(\Phi(\lambda_i),\Phi(\lambda_{\m i}))$.
\end{lemma}

\begin{proof}
  If $\ut_1^G(\lambda_1,\lambda_2)\leq B$, then
  $\outcome^G(\lambda_1,\lambda_2)$ is winning, and we let 
  $\pi = s_0\dots s_n = \outcome^{G,\target_1}(\lambda_1,\lambda_2)$.
  We let $\pi' = (\pi_0)\dots (\pi_n)$ where for all $k\leq n$, 
  $\pi_k = s_0\dots s_k$. Since $\pi$ is
  bounded, $\pi_k\in \play_{B^G}(G')$. By definition of $\Phi$ we
  clearly have $\pi' =
  \outcome^{G'}(\Phi(\lambda_1),\Phi(\lambda_2))$, from
  which we get $\ut_1^G(\lambda_1,\lambda_2) = 
  \ut_1^{G'}(\Phi(\lambda_1),\Phi(\lambda_n))$.

  If $\outcome^{G'}(\Phi(\lambda_1),\Phi(\lambda_2))$ is winning for
  Player $1$, we let $\pi' =
  \outcome^{G',\target'_1}(\Phi(\lambda_1),\Phi(\lambda_2))$ and $\pi
  = \last(\pi_0)\dots \last(\pi_n)$ where $\pi' = \pi_0\dots \pi_n$.
  Clearly, $\pi$ is
  a winning play for Player $1$ in $G$ and by definition of $\Phi$,
  $\pi = \outcome^{G,\target_1}(\lambda_1,\lambda_2)$, from which we
  get the equality of the utilities.

\end{proof}

\begin{lemma}\label{lem:mapbr}
  For all $i=1,2$, all $j\geq 0$. If $\Phi(P_i^j(G)) = P_i^j(G')$ and
  there is a strategy $j$-bounded by $b^G$
  in $P_i^j(G)$, then for all $\lambda_{\m i}\in P_{\m i}^j(G)$, 
  $\br^{G,P_i^j(G)}_i(\lambda_{\m i}) =
  \br^{G',P_i^j(G')}_i(\Phi(\lambda_{\m i}))$.
\end{lemma}

\begin{proof}
  Let $\eta_i\in P_i^j(G)$ be a strategy $j$-bounded by $b^G$ (it
  exists by hypothesis), and let $\lambda_{\m i}\in P_{\m i}^j(G)$.
  Since $\eta_i$ is $j$-bounded, it is $j$-winning and by
  Lemma \ref{lem:mapbounded}, $\ut_i^G(\eta_i,\lambda_{\m i}) =
  \ut_i^{G'}(\Phi(\eta_i),\Phi(\lambda_{\m i}))$. Therefore
  $\ut_i^{G'}(\Phi(\eta_i),\Phi(\lambda_{\m i}))<+\infty$.

  Let $\lambda_i\in P_i^j(G)$ which minimizes
  $\br^{G,P_i^j(G)}_i(\lambda_{\m i})$. Let $\pi =
  \outcome^{G,\target_i}(\lambda_i,\lambda_{\m i})$. 
  We have $\mu_i(\pi) \leq \ut_i^G(\eta_i,\lambda_{\m
    i}) \leq b^G$. Since the weights are strictly positive integers,
  $|\pi| \leq b^G$, and therefore $\mu_{\m i}(\pi) \leq b^G M^G =
  B^G$. By Lemma \ref{lem:mapbounded}, we get
  $\ut_i^G(\lambda_i,\lambda_{\m i}) =
  \ut_i^{G'}(\Phi(\lambda_i),\Phi(\lambda_{\m i})$. Therefore
  $\ut_i^G(\lambda_i,\lambda_{\m i}) = \br^{G,P_i^j(G)}_i(\lambda_{\m i}) \geq
  \br^{G',P_i^j(G')}_i(\Phi(\lambda_{\m i}))$.

  Conversely, let $\lambda_i'\in P_i^j(G')$ which minimizes
  $\br^{G',P_i^j(G')}_i(\Phi(\lambda_{\m i}))$. Therefore $\ut_i^{G'}(\lambda'_i,\Phi(\lambda_{\m i})) \leq
  \ut_i^{G'}(\Phi(\eta_i),\Phi(\lambda_{\m i}))<+\infty$. 
  Since $\Phi(P_i^j(G)) = P_i^j(G)$ by hypothesis, there exists 
  $\lambda_i\in P_i^j(G)$ such that $\Phi(\lambda_i) = \lambda'_i$.
  Since $\ut_i^{G'}(\lambda'_i,\Phi(\lambda_{\m i}))$ is finite,
  $\outcome^{G'}(\lambda'_i,\Phi(\lambda_{\m i}))$ is winning for
  Player $i$. It is easy to see that $\outcome^{G,\target_i}(\lambda_i,\lambda_{\m
    i}) = \last(\pi_0)\dots\last(\pi_n)$ where $\pi_0\dots \pi_n=
  \outcome^{G',\target_i'}(\lambda'_i,\Phi(\lambda_{\m
    i}))$ and that they both have the same utility, i.e. 
  $\ut_i^G(\lambda_i,\lambda_{\m i}) =
  \ut_i^{G'}(\lambda'_i,\Phi(\lambda_{\m i}))$. Since
  $\ut_i^{G'}(\lambda'_i,\Phi(\lambda_{\m i}))  =
  \br^{G',P_i^j(G')}_i(\Phi(\lambda_{\m i}))$, we get 
  $\br^{G',P_i^j(G')}_i(\Phi(\lambda_{\m i}))\geq
  \br^{G,P_i^j(G)}_i(\lambda_{\m i})$.
\end{proof}

\begin{lemma}\label{lem:equivbounded}
  For all $j\geq 0$, $i=1,2$, $\lambda_i\in P_i^j(G)$.
  If $\Phi(P_{\m i}^j(G)) = P_{\m i}^j(G')$, then
  $\lambda_i$ is $j$-bounded by $B^G$ iff $\Phi(\lambda_i)$ is $j$-winning. 
  If $\lambda_i$ is $j$-bounded by $B^G$, then
  for all $\lambda_{\m i}\in S^{\m i}_j(G)$,
  $\ut_i^G(\lambda_i,\lambda_{\m i}) =
  \ut_i^{G'}(\Phi(\lambda_i),\Phi(\lambda_{\m i}))$.
\end{lemma}

\begin{proof}
  If $\lambda_i\in \Lambda_i(G)$ is $j$-bounded by $B^G$. Then let
  $\lambda'_{\m i}\in P_{\m i}^j(G')$. Since $\Phi(P_{\m i}^j(G)) =
  P_{\m i}^j(G')$, there exists $\lambda_{\m i}\in P_{\m i}^j(G)$ such
  that  $\Phi(\lambda_{\m i}) = \lambda_{\m i}'$. 
  Since $\lambda_i$ is $j$-bounded, we are in the condition of Lemma
  \ref{lem:mapbounded}, therefore $\ut_i^G(\lambda_i,\lambda_{\m i}) =
  \ut_i^{G'}(\Phi(\lambda_i),\lambda'_{\m i}) < +\infty$. Therefore
  $\Phi(\lambda_i)$ wins against $\lambda'_{\m i}$.

  Conversely, if $\Phi(\lambda_i)$ is $j$-winning, then let
  $\lambda_{\m i}\in P_{\m i}^j(G)$. By hypothesis, $\Phi(\lambda_{\m
    i})\in P_{\m i}^j(G')$. Therefore $\Phi(\lambda_i)$ wins against
  $\Phi(\lambda_{\m i})$. Let $\pi_0\dots \pi_n  =
  \outcome^{G',\target'_i}(\Phi(\lambda_i),\Phi(\lambda_{\m i}))$. 
  Clearly, by definition of $\Phi$ and $G'$, $\lambda_i$ wins
  against $\lambda_{\m i}$ and
  $\outcome^{G,\target_i}(\lambda_i,\lambda_{\m i}) =
  \last(\pi_0)\dots \last(\pi_n)$.
\end{proof}

\begin{lemma}\label{lem:induction}
  $\forall j\!\geq\! 0$, $\forall i=1,2$, if $\Phi(P_i^j(G)) =
  P_j^j(G')$ then:

\begin{itemize}

\item[$(i)$] $\reg_i^{G,j+1} = +\infty$ iff $\reg_{i}^{G',j+1} = +\infty$

\item[$(ii)$]  $\forall \lambda_i\in P_i^{j+1}(G)\cup \Phi^{-1}(P_i^{j+1}(G'))$, 

    $\reg^{G,j+1}_i(\lambda_i) = \reg^{G',j+1}_{i}(\Phi(\lambda_i))$
    
\item[$(iii)$] $\Phi({P_i^{j+1}(G)}) = P_i^{j+1}(G')$

\end{itemize}
\end{lemma}

\begin{proof}
\textbf{(\textit{i})} If $\reg_i^{G,j+1} < +\infty$, then it means that
  there is a $j$-winning strategy $\lambda_i\in P_j^i(G)$. By Lemma
  \ref{lem:equivbounded}, $\Phi(\lambda_i)$ is $j$-winning. 
  Since by hypothesis, $\Phi(P_j^i(G)) = P_j^i(G')$, 
  $\Phi(\lambda_i)\in P_j^i(G')$, so that $\reg_i^{G',j+1} \leq
  \reg_i^{G',j+1}(\Phi(\lambda_i)) < +\infty$.

  Conversely, if $\reg_i^{G',j+1} < +\infty$, there is a $j$-winning
  strategy $\lambda'_i\in P_j^i(G')$. By hypothesis, 
  there is $\lambda_i\in P_j^i(G)$ such that $\Phi(\lambda_i) =
  \lambda'_i$, and by Lemma \ref{lem:equivbounded}, $\lambda_i$ is
  $j$-bounded, and in particular $j$-winning. Therefore
  $\reg_i^{G,j+1} < +\infty$.

\textbf{(\textit{ii})} The proof is in two parts, depending on whether
$\reg_i^{G,j+1}$ is finite or not. 

\vspace{2mm}

\textbf{\textit{(ii).a}} If $\reg_i^{G,j+1} = +\infty$, then by $(i)$,
$\reg_i^{G',j+1} = +\infty$.

Let $\lambda_i\in P_i^{j+1}(G)$. 
Since $P_i^{j+1}(G)\subseteq P_j^i(G)$, $\lambda_i\in P_j^i(G)$. 
By hypothesis, $\Phi(P_j^i(G)) = P_j^i(G')$, therefore
$\Phi(\lambda_i)\in P_j^i(G')$. Therefore
$\reg_i^{G',j+1}(\Phi(\lambda_i)) = +\infty =
\reg_i^{G,j+1}(\lambda_i)$.

If $\lambda_i\in \Phi^{-1}(P_i^{j+1}(G'))$, then 
$\reg_i^{G,j+1}(\Phi(\lambda_i)) = +\infty$.
Since $P_i^{j+1}(G')\subseteq P_j^i(G')$, 
$\Phi(\lambda_i)\in P_j^i(G')$. By hypothesis, 
$\Phi(P_j^i(G)) = P_j^i(G')$, therefore $\lambda_i\in P_j^i(G)$, and
$\reg_i^{G,j+1}(\lambda_i) = +\infty =
\reg_i^{G,j+1}(\Phi(\lambda_i))$, since $\reg_i^{G,j+1} = +\infty$.

\textbf{\textit{(ii).b}} If $\reg_i^{G,j+1} < +\infty$, then by $(i)$, 
$\reg_i^{G',j+1} < +\infty$. Let $\lambda_i\in P_i^{j+1}(G)\cup
\Phi^{-1}(P_i^{j+1}(G'))$. We prove that for all $\lambda_{\m i}\in
P_{\m i}^j(G)$,

\begin{enumerate}
  \item $\ut_i^G(\lambda_i,\lambda_{\m i}) =
    \ut_i^{G'}(\Phi(\lambda_i),\Phi(\lambda_{\m i})$

  \item $\br^{G,P_j^i(G)}_i(\lambda_{\m i}) = \br^{G',P_j^i(G')}_i(\Phi(\lambda_{\m i}))$
\end{enumerate}

For $1$, If $\lambda_i\in P_i^{j+1}(G)$, then since $\reg_i^{G,j+1}< +\infty$,
by Lemma \ref{lem:boundedstrat}, $\lambda_i$ is $j$-bounded by $b^G$,
and \textit{a fortiori} by $B^G$. By Lemma \ref{lem:equivbounded}, we
get the result. If $\lambda_i\in \Phi^{-1}(P_i^{j+1}(G'))$, then
since $\reg_i^{G',j+1}<+\infty$, $\Phi(\lambda_i)\in P_{i}^{j}(G')$
is necessarily winning. By Lemma \ref{lem:equivbounded}, $\lambda_i$
is $j$-bounded by $B^G$ and again by the same lemma, we get the
result.

For $2$, since $\reg_{i}^{G,j+1}<+\infty$, by Lemma
\ref{lem:boundedstrat}, there is a strategy of Player $i$ $j$-bounded
by $B^G$ in $P_j^i(G)$. Therefore we can apply Lemma \ref{lem:mapbr}
and we get the result.

Finally we have:

    $$
    \begin{array}{lllllll}
        & \reg_i^{G,j+1}(\lambda_i) \\
        = & \max\limits_{\lambda_{\m i}\in P_{\m i}^j(G)} [\ut_i^G(\lambda_i,\lambda_{\m i}) -
          \br^{G,P_j^i(G)}_i(\lambda_{\m i})] \\

        = & \max\limits_{\lambda_{\m i}\in P_{\m i}^j(G)} [\ut_i^{G'}(\Phi(\lambda_i),\Phi(\lambda_{\m i})) -
          \br^{G,P_j^i(G)}_i(\Phi(\lambda_{\m i}))] \\

        & \text{ (by $(1)$ and $(2)$)} \\

        = & \max\limits_{\lambda'_{\m i}\in P_{\m i}^j(G')} [\ut_i^{G'}(\Phi(\lambda_i),\lambda'_{\m i}) -
          \br^{G',P_j^i(G')}_i(\lambda'_{\m i})] \\

        & \text{ (since $\Phi(P_j^i(G)) = P_j^i(G')$ by hypothesis)} \\

        = &  \reg_i^{G',j+1}(\Phi(\lambda_i)) \\
 
    \end{array}
    $$

    \vspace{2mm}

    \textbf{(iii)} Let $i\in\{1,2\}$ and $\lambda_i\in S^i_{j+1}(G)$.
    Suppose that $\Phi(\lambda_i)\not\in S^i_{j+1}(G')$. It means that $\Phi(\lambda_i)$ does not minimize
    the $j+1$-th regret in $G'$. Therefore there exists another strategy $\lambda_i'\in P_i^{j+1}(G')$ such that
    $\reg_i^{G',j+1}(\lambda_i') < \reg_i^{G,j+1}(\Phi(\lambda_i))$.
    By $(i)$, we get $\reg_i^{G,j+1}(\gamma_i) < \reg_i^{G,j+1}(\lambda_i)$, for all $\gamma_i\in \Phi^{-1}(\lambda_i')$.
    Since $P_i^{j+1}(G')\subseteq P_{i}^{j}(G')$ and $\Phi(P_j^i(G)) = P_j^i(G')$, we have $\Phi^{-1}(\lambda'_i)\subseteq P_j^i(G)$, and we
    get a contradiction on the minimality of $\lambda_i$.

    Conversely, let $\lambda_i'\in P_i^{j+1}(G')$. Suppose that $\lambda'_i\not\in \Phi(P_i^{j+1}(G))$. Since $\lambda'_i\in P_{i}^{j}(G')$, by hypothesis,
    $\lambda'_i\in \Phi(P_j^i(G))$. Therefore there exists $\lambda_i\in P_j^i(G)$ such that $\Phi(\lambda_i) = \lambda'_i$, but $\lambda_i\not\in P_i^{j+1}(G)$. 
    It means that $\lambda_i$ did not survive to the $j$-th iteration. In other words, for all strategy $\gamma_i\in P_i^{j+1}(G)$, $\reg_i^{G,j+1}(\gamma_i) < \reg_i^{G,j+1}(\lambda_i)$.
    Since $\lambda_i\in\Phi^{-1}(P_i^{j+1}(G'))$, by $(ii)$ we have $\reg_i^{G,j+1}(\lambda_i) = \reg_i^{G',j+1}(\Phi(\lambda_i)) = \reg_i^{G',j+1}(\lambda'_i)$. By $(ii)$, we also
    have $\reg_i^{G,j+1}(\gamma_i) = \reg_i^{G',j+1}(\Phi(\gamma_i))$. Therefore $\reg_i^{G',j+1}(\Phi(\gamma_i)) < \reg_i^{G',j+1}(\lambda'_i)$.
    Since $\gamma_i\in P_j^i(G)$ and by hypothesis, $\Phi(P_j^i(G)) = P_j^i(G')$, we have $\Phi(\gamma_i)\in P_j^i(G')$. Therefore we get a strategy
    $\Phi(\gamma_i)$ of $P_j^i(G')$ with a lower $(j+1)$-th regret than the $(j+1)$-regret of $\lambda'_i$. This is in contradiction with $\lambda'_i\in P_i^{j+1}(G')$.
\end{proof}

\paragraph{Lemma \ref{lem:corres}}

\begin{proof}
  Clearly, $\Phi(\Lambda_i(G)) = \Lambda_i(G')$. Therefore we can
  apply Lemma \ref{lem:induction} (proved in Appendix) so that items $(i)$, $(ii)$ and
  $(iii)$ holds at rank $0$. In particular, $\Phi(P_i^1(G)) =
  P_i^1(G')$. Therefore we can again apply Lemma \ref{lem:induction}
  at rank $1$. More generally, for all $j\geq 1$, we have:

  \begin{enumerate}
    \item for all $\lambda_i\in P_j^i(G)$, $\reg_i^{G,j}(\lambda_i) =
      \reg_{i}^{G',j}(\Phi(\lambda_i))$;
    \item $\Phi(P_j^i(G)) = P_j^i(G')$.
  \end{enumerate}
\end{proof}

\end{document}